\documentclass[12pt]{article}
\usepackage{amsmath,amssymb}
\usepackage{graphicx}
\usepackage[dvips]{epsfig} 
\usepackage{amsmath,bm}
\usepackage[center]{subfigure}
\usepackage{epstopdf}
\usepackage{times}
\usepackage{biograph}
\usepackage{cite}
 \usepackage{mathptmx}
\usepackage{latexsym}
\usepackage[latin1]{inputenc}
\usepackage[T1]{fontenc}
\usepackage{amsfonts}
\usepackage{amsthm}
\usepackage[mathscr]{eucal}
\usepackage[english,francais]{babel}
\usepackage{verbatim}
\usepackage{Times}
\usepackage{appendix}
\usepackage{color}

 \newcommand{\ds}{\displaystyle}

\newtheorem{thm}{Theorem}[section]

\newtheorem{cor}[thm]{Corollary} 
\newtheorem{prn}[thm]{Proposition}
\newtheorem{rem}{Remark}[section]
\newtheorem{ex}{Example}[section]
\numberwithin{equation}{section} 
\numberwithin{thm}{section}
\numberwithin{rem}{section}
\numberwithin{ex}{section}
\begin{document}
\title{On Combining Estimation Problems Under Quadratic Loss: A Generalization}
\date{}
\author{S\'ev\'erien Nkurunziza{\thanks {University of Windsor, 401 Sunset
Avenue, Windsor, Ontario, N9B 3P4. Email: severien@uwindsor.ca}}}
\thispagestyle{empty} \selectlanguage{english} \maketitle
\maketitle\thispagestyle{empty}
\begin{abstract} The main theorem in Judge and
Mittelhammer [Judge, G. G., and Mittelhammer,~R.~(2004), A
Semiparametric Basis for Combining Estimation Problems under
Quadratic Loss; JASA, {\bf 99}, 466, 479--487] stipulates that, in
the context of nonzero correlation, a sufficient condition for the
Stein rule~(SR)-type estimator to dominate the base estimator is
that the dimension $k$ should be at least 5. Thanks to some refined
inequalities, this dominance result is proved in its full
generality; for a class of estimators which includes the SR
estimator as a special case. 
Namely, we prove that, for any member of the derived class,
$k\geqslant 3$ is a sufficient condition regardless of the
correlation factor. We also relax the Gaussian condition of the
distribution of the base estimator, as we consider the family of
elliptically contoured variates. Finally, we waive the condition on
the invertibility  of the variance-covariance matrix of the base and
the competing estimators. Our
theoretical findings are corroborated by some simulation studies, and
the proposed method is applied to the Cigarette dataset.

\end{abstract}
 \pagenumbering{arabic} \addtocounter{page}{0}
\section{Introduction and statistical model }\label{sec:intro}
\subsection{Introduction}
The multiple regression model is a common statistical tool for
investigating the relationship between a response variable and
several explanatory variables. One of the main issues in regression
analysis consists in estimating the regression coefficients. In
particular, in the context of a linear regression model, it is
common to use the ordinary least squares estimator~(OLSE). Indeed,
under the normality of the errors term, OLSE is known to be the
maximum likelihood estimator as well as the minimum variance
unbiased estimator. However, in case some prior information~(from
outside the sample) is available, OLSE may not be optimal. For
instance, this prior information may be due to past statistical
investigations, when these investigations could have concluded that
some regression coefficients are not statistically significant.
Another source of prior information may be the expertise in a
certain field, which establishes an association between the
regressor variables. Such a situation arises in economic theory
where, for example, it is common to consider that the sum of the
exponents in a Cobb-Douglas production~(see Douglas and Cobb,~1928)
is equal to one. 

From the statistical inference point of view, it is important to
incorporate the available prior information in the estimation method
in order to improve upon the OLSE. For instance, if such prior
information can be expressed in the form of exact linear
restrictions binding the regression coefficients, instead of using
the OLSE, one can resort to a competing estimator which is also
known as the restricted least squares estimator~(RLSE); it is known
that the RLSE dominates the OLSE in such cases. In the sequel, the
OLSE will be referred to as the base estimator while the RLSE will
be referred to as the restricted estimator or the competing
estimator. Thus, in the case where some exact prior information is
available, the practitioners should use the restricted estimator in
order to estimate the target parameter while if only the sample
information is
available, the base estimator is to be preferred. 

Nevertheless, in some circumstances, the prior information is nearly
correct and thus, we want to incorporate an additional information
but we are not completely sure about it. Such uncertainty about the
additional information may be induced by a change in the phenomenon
underlying the regression model. Another context is the one where
the prior information comes from experts in a field, the uncertainty
reflects the imprecision in the experts' information or judgements.
In the case where the prior information is that, from the past
statistical investigations, some regression coefficients are not
statistically significant, the uncertainty may reflect the fact that
a field specialist believes that
the nonsignificant explanatory variables are important.

In these cases, we have to choose how to incorporate uncertain prior
information into the inference procedure. Technically, in order to
use both the sample and the uncertain prior, we can combine the base
estimator and the restricted estimator and thus it is important to
find an optimal combination. In the context of the linear regression
model, Judge and Mittelhammer~(2004) proposed a Stein-type estimator
and derived a sufficient condition for the risk dominance of
Stein-type estimator relative to a certain base estimator. However,
the main result, in Judge and Mittelhammer~(2004)[JM], has some
limitations. First, the error term is supposed to be normally
distributed. Second, the variance-covariance matrix of the joint
distribution of the base estimator and the competing estimator is
supposed to be invertible. This last assumption excludes, for
example, a case where the prior information is about the
non-significance of some regression coefficients. Third, the derived
sufficient condition is too restrictive in the sense that it
excludes the case of a multiple regression model with less than five
regressors. Thus, the condition in JM~(2004) is not applicable to
the cases of quadratic or cubic regression models. However, in many
applications (see Ashton {\em et al.},~2008,
Fernandez-Juricic {\em et al.},~2003, among others), if a linear fit
is not appropriate, a quadratic or cubic regression proves to be a
simple and an adequate model. The last example is the Cigarette dataset
produced by the USA Federal Trade Commission which can be found in
Mendenhall  and Sincich~(1992). For this data set, the method in
JM~(2004) is not applicable since  we have only three
explanatory variables. In Section~\ref{sec:simuldata}, we analyse
this dataset and we show that our method performs very well.

 In this paper, we generalize in four ways the main result in
JM~(2004) which gives a sufficient condition for the risk dominance
of Stein-type estimator relative to a certain base estimator. First,
we present a class of estimators which includes as a special case
the Stein rule-type estimator given in JM~(2004). Second, we relax
the condition on the dimension of the parameter space. Third, we
waive the condition on the invertibility of the variance-covariance
matrix of the base estimator and the competing estimator. Thus, the
proposed methodology works also in the case where the practitioners
suspect some linear restrictions binding the regression
coefficients. Fourth, we extend the main result to the case of a
family of elliptically contoured distributions. To this end, recall
that the normal distribution is a member of the elliptically
contoured distributions, and, as explained in Provost and
Cheong~(2000), many test statistics and optimality properties
underlying Gaussian random samples remain unchanged for elliptically
contoured random samples. For further discussions and advantages of
elliptically contoured distributions, we refer for example to
Abdous~\emph{et~al.}~(2004), Liu~\emph{et al.}~(2009) and
references therein. Finally, the main key for establishing our results
consists in deriving some inequalities and bounds which are more refined than that
used in JM~(2004).

The remaining of this paper is organized as follows.
Section~\ref{sec:statmodel} presents the statistical model which is
given in JM~(2004) as well as the highlights of  our contributions.
In Section~\ref{sec:classSR}, we present a class of Stein rule-type
of estimators and their risk function. Section~\ref{sec:mainres}
gives the main results of this paper in the Gaussian  case and, more
generally, in the elliptically contoured random case. We also show,
in Section~\ref{sec:mainres}, that the proposed method works in the
context where the variance-covariance of the base estimator and the
competing estimator is singular.  In Section~\ref{sec:simuldata}, we
present some simulation results for small sample sizes  as well as
an analysis of a real data set. Section~\ref{sec:conclu} gives some
concluding remarks. For the convenience of the reader, technical
proofs are given in the Appendix. 

\subsection{Statistical model and main
contributions}\label{sec:statmodel} In this section, we recall the
statistical model and the assumptions as well as some preliminary
results which are given in JM~(2004). Thus, this section presents
only the model for which the error term is normally distributed. As
mentioned in the Introduction, this is a preliminary step as we show
later that 
the result established under the normality assumption holds also in
the cases of elliptically contoured variables.

Following JM~(2004), we consider the estimation problem of a
$k$-dimensional location parameter vector when one observes an
$n$-dimensional sample vector $y$ such that
$y=\bm{X}\beta+\epsilon$, where $\bm{X}$ is an $n\times k$ design
matrix of rank $k$ and $\epsilon$ is an $n$-dimensional random
vector such that $\textrm{E}(\epsilon)=0$ and
$\textrm{cov}(\epsilon)=\sigma^{2}\bm{I}_{n}$. Further, as in the
quoted paper, we consider the scenario where there exists some
uncertainty concerning the above statistical model, which leads to
uncertainty concerning the appropriate inference method. For more
details about these issues, we refer, for example, to JM~(2004),
Saleh~(2006), Hossain~{\em et
al.}~(2009), Morris and Lysy~(2012) among others.

In the case where the above statistical model is appropriate, it is
natural to estimate the target parameter $\beta$ by using the
least-squares estimator (LS)
$\hat{\delta}^{LS}=(\bm{X}'\bm{X})^{-1}\bm{X}'y$. Further, in the
context of an alternative statistical model, one can consider the
competing estimator $\tilde{\beta}$, which is such that
$\textrm{E}(\tilde{\beta})=\beta+\gamma$,
$\textrm{cov}(\tilde{\beta})=\bm{\Phi}$,
$\textrm{cov}(\hat{\beta},\tilde{\beta})=\bm{\Sigma}$. Thus, as in
JM~(2004), the two estimators $\hat{\beta}$ and $\tilde{\beta}$ are
assumed to be correlated, and $\tilde{\beta}$ may be biased with
bias $\gamma$. In the context of uncertainty about which one of  the
two statistical models is more appropriate, it is common to consider
an estimator which combines the two estimators in an optimal way.
Originally, this type of method was introduced by James and
Stein~(1961). Over the last 50 years, numerous papers have been
written around the topic so that it would be impossible to summarize
all of them. To give some closely related references, we mention
Bock~(1975), Judge and Bock~(1978), JM~(2004), Saleh~(2006), Nkurunziza and Ahmed~(2010),
Nkurunziza~(2011), and Tan~(2015) and references therein.

In our paper, we extend the following: JM~(2004) stipulate that (see
their main theorem), in the case of nonzero correlation between the
base estimator $\hat{\beta}$ and the alternative estimator
$\tilde{\beta}$, a sufficient condition for the Stein rule~(SR)-type
estimator to dominate the base estimator is $k\geqslant 5$.

In this paper, we extend this result in four ways. First, we
construct a class of estimators which includes as a special case the
SR estimator given in JM~(2004). Second, we prove that,
regardless of the presence of correlation, the condition $k\geqslant
3$ remains sufficient for any estimator of the proposed class of SR
estimators to dominate in mean squared error the base estimator. The
impact of this finding consists in the fact that, unlike the
result in JM~(2004), the established method can be applied to the
case where the number of regressors is less than five as, for
example, the case of a quadratic or a cubic regression model. Third,
we also generalize the method in JM~(2004) to the case where the
joint distribution of the base estimator and the restricted
estimator may be singular. This last result can be very useful in
the  case where the statistician suspects some linear restrictions
binding the regression coefficients. This includes, for example, the
case where the prior information from past statistical
investigations is that some regression coefficients are not
statistically significant, while the expert in the field of
application believes that the corresponding explanatory variables
should be in the model. Fourth, we prove that the established
results hold if the normality assumption is replaced with that of
elliptically contoured variates. Technically, in order to
derive our findings, we  establish some inequalities
which are more refined than that in JM~(2004). Finally, let us note that the
simulation results are in agreement with the above theoretical
findings. More specifically, the simulations show that the risk
dominance of some SR estimators increases as the correlation
increases.
\section{A class of  Stein rule estimators and the risk function}\label{sec:classSR}
\subsection{A class of  Stein rule-type estimators}
In this subsection, we present a class of  Stein
rule~(SR)-type estimators which includes as a special case the SR estimator in JM~(2004).
First, recall that the results given in this paper hold under a very general statistical
model than that in JM~(2004). More precisely, the established results hold whenever the
estimators $\hat{\beta}$ and $\tilde{\beta}$ follow jointly an elliptically contoured
distribution. 
First, suppose that the estimators $\hat{\beta}$
and $\tilde{\beta}$ are jointly Gaussian. 
Thus, let
\begin{eqnarray}
\left(
            \begin{array}{c}
              \hat{\beta}-\beta \\
              \tilde{\beta}-\beta \\
            \end{array}
          \right)
\sim\mathcal{N}_{2k}\left((\bm{0},\bm{\gamma}')',\left(
                                                   \begin{array}{cc}
                                                     \bm{A} & \bm{\Sigma} \\
                                                     \bm{\Sigma}' & \bm{\Phi} \\
                                                   \end{array}
                                                 \right)\label{distrbeta}
\right),
\end{eqnarray}
where, as in JM~(2004), the matrices $\bm{A}$, $\bm{\Phi}$,
$\bm{\Xi}=\bm{A}-\bm{\Sigma}-\bm{\Sigma}'+\bm{\Phi}$ are assumed to
be positive definite. This assumption will be waived in Subsection~\ref{sec:extension}
to study the case where the matrices $\bm{\Xi}$ and $\bm{\Phi}$ may be singular.
Further, let $c$ be real number and let $h$ be real-valued measurable and square-integrable (with respect to the Gaussian measure). We consider the following class of SR estimators
\begin{eqnarray}
\hat{\beta}^{S}(h,c)=\hat{\beta}+c\,\, h\left(\hat{\beta},\tilde{\beta}\right)
\left(\hat{\beta}-\tilde{\beta}\right). \label{spsl}
\end{eqnarray}
\begin{ex}
\begin{enumerate}
  \item For a given $m$-column vector $e$,let $\|e\|^{2}=e'e=\textrm{\em trace}(ee')$. If $h\left(\hat{\beta},\tilde{\beta}\right)=1/\|\hat{\beta}-\tilde{\beta}\|^{2}$, for a known real number $c$, from~\eqref{spsl}, we get the estimator
\begin{eqnarray}
\hat{\delta}\left(\hat{\beta},\tilde{\beta};c\right)
=\hat{\beta}-\ds{\frac{c}{\|\hat{\beta}-\tilde{\beta}\|^{2}}}
\left(\hat{\beta}-\tilde{\beta}\right),\label{SPSL}
\end{eqnarray}
that is the Stein rule~(SR)-type estimator given in~JM~(2004).
  \item Let $\hat{a}=S^{2}\textrm{\em trace}((\bm{X}'\bm{X})^{-1})-\textrm{\em trace}(\hat{\bm{\Sigma}})$,
$S^{2}=(n-k)^{-1}\|y-\bm{X}\hat{\beta}\|^{2}$, $\hat{\bm{\Sigma}}$
is an unbiased and/or a consistent estimator for $\bm{\Sigma}$. If
$c=-\hat{a}$,
$h\left(\hat{\beta},\tilde{\beta}\right)=1/\|\hat{\beta}-\tilde{\beta}\|^{2}$,
the estimator in~\eqref{spsl} becomes the Semiparametric Stein-Like
(SPSL) estimator given in JM~(2004). Namely, we have
\begin{eqnarray}
\hat{\beta}^{S}=\hat{\beta}-\ds{\frac{\hat{a}}{\|\hat{\beta}-\tilde{\beta}\|^{2}}}
(\hat{\beta}-\tilde{\beta}). \label{spsla}
\end{eqnarray}
  \item If $h\left(\hat{\beta},\tilde{\beta}\right)=0$, the estimator in~\eqref{spsl} yields the base estimator $\hat{\beta}$.
  \item If $c=-1$, $h\left(\hat{\beta},\tilde{\beta}\right)=1$, we have $\hat{\beta}^{S}(-1,1)=\tilde{\beta}$.
\end{enumerate}
\end{ex}
As an important point, in this paper, the random quantity $h\left(\hat{\beta},\tilde{\beta}\right)$ is a statistic in the sense that it can be computed whenever we have the observations. As for the real value $c$ which is assumed to be known in \eqref{spsl}, this is similar to that used in SR-estimator in~JM~(2004). The impact of replacing $c$ by its corresponding consistent estimator should be similar to that in~JM~(2004). In practice, the value of $c$ can be obtained by using a re-sampling technique as the bootstrap.
\subsection{Risk function}
The performance of the proposed class of estimators is studied under the quadratic loss function.
Thus, the quadratic risk function, so-called the mean squared error~(MSE),
 of the class of the estimators in~\eqref{spsl} is
$$\textrm{MSE}\left(\hat{\beta}^{S}(h,c)\right)=
\textrm{E}\left[\left(\hat{\beta}^{S}(h,c)-\beta\right)'
\left(\hat{\beta}^{S}(h,c)-\beta\right)\right].$$
then, by using~\eqref{distrU}
and~\eqref{zandr}, we have
\begin{eqnarray}
\textrm{MSE}\left(\hat{\beta}^{S}(h,c)\right)
=\textrm{trace}(\bm{A})-2c\eta(h)+c^{2}\omega(h),\label{mse}
\end{eqnarray}
where
\begin{eqnarray}
\eta(h)=\textrm{E}\left[h\left(\hat{\beta},\tilde{\beta}\right)\left(\hat{\beta}-\beta\right)'\left(\hat{\beta}-\tilde{\beta}\right)\right],
\quad{ }
\omega(h)=\textrm{E}\left[h^{2}\left(\hat{\beta},\tilde{\beta}\right)
\left\|\hat{\beta}-\tilde{\beta}\right\|^{2}\right],
\label{etaomegah}
\end{eqnarray}
assuming that these expectations are defined. Thus, from~\eqref{mse}, it is obvious that,
for all
$c\in\left(\min\{0,2\eta(h)/\omega(h)\},\,
\max\{0,2\eta(h)/\omega(h)\}\right)$,  we have   
$$\textrm{MSE}\left(\hat{\beta}^{S}(h,c)\right)
<\textrm{MSE}(\hat{\beta}),$$
provided that $\eta(h)$ and $\omega(h)$ exist. Further, from~\eqref{mse}, one concludes that the optimal choice of
$c$ is $c_{*}=\eta(h)/\omega(h)$ and thus,
\begin{eqnarray}
\textrm{MSE}\left(\hat{\beta}^{S}(h,c^{*})\right)
=\textrm{trace}(\bm{A})-(\eta^{2}(h)/\omega(h)).\label{mseopt}
\end{eqnarray}
\begin{rem}\label{rem:opt}
Assuming that $\eta(h)$ and $\omega(h)$ are defined, \eqref{mseopt} implies that, for a fixed value of $\omega(h)$, the MSE
of the SR estimator decreases as $\eta(h)$ increases. Further, for a fixed value
of $\eta(h)$, the MSE of the SR estimator decreases as $\omega(h)$ decreases. The
simulation results given in Section~\ref{sec:simulat} are in
agreement with this analysis. To make this idea more precise, we
first note that, for a given $h$, the values of $\omega(h)$ and $\eta(h)$ depend both on
$\gamma$, the bias of the estimator $\tilde{\beta}$, and on the
covariance between the estimator $\tilde{\beta}$ and $\hat{\beta}$.
In particular, the simulation results show that the risk dominance
of the SPSL increases as the correlation increases. Also, the
simulation results show that the risk dominance of the SPSL
decreases as the norm of the bias increases.
\end{rem}
As mentioned above,
the derivation of the MSE in~\eqref{mse} and \eqref{mseopt}
assumes the existence of $\omega(h)$ and $\eta(h)$. Thus, it is
important to derive the conditions under which these expectations are defined.
To this end, we require that the function $h$ satisfies the following assumption.

\noindent {\bf \underline{Assumption $(\mathcal{H}_{1})$}}: {\em The
function $h$ is such that $\|x-y\|^{2}|h(x,y)|$ is bounded i. e.
$$|h(x,y)|=O\left(\|x-y\|^{-2}\right).$$}
\begin{rem}
It should be noticed that the function $h$ which gives the SR estimator  satisfies
the above assumption. Indeed, in this case, we have $\|x-y\|^{2}|h(x,y)|=1$.
Also, the function $h\equiv 0$ which gives the base estimator
satisfies the above assumption. As another example of a function $h$ which satisfies the Assumption $(\mathcal{H}_{1})$,
one can take $h(x,y)=\ds{\frac{1}{1+\|x-y\|^{p}}}$ for some $p\geqslant2$.
\end{rem}
Below, we prove that, under the above assumption, regardless of the presence of correlation,
the condition $k\geqslant 3$ remains sufficient for any estimator of the class in~\eqref{spsl} to
dominate in mean square error the base estimator. In particular, since the SR estimator is a member
of the class of the estimators in~\eqref{spsl}, the established result proves that, regardless of the presence of correlation, the condition $k\geqslant 3$ remains sufficient for the SR estimator to
dominate in mean square error the base estimator. We also prove that this conclusion holds if the normality assumption is replaced by that of elliptically contoured variates.

\section{Main results}\label{sec:mainres}
In this section, we present the main results of this paper. As an
intermediate step, we derive below three propositions and a theorem
which play a central role in deriving the main result. In summary, these
results are useful in deriving a more refined inequality than that used in
JM~(2004). In order to
simplify the presentation of the main results, we define some
notations which will be used for the remaining of the paper. 
Let $U=(U'_{1},U'_{2})'$ where $U_{1}=\hat{\beta}-\beta$ and $U_{2}=
\tilde{\beta}-\beta$. From~\eqref{distrbeta}, we have
\begin{eqnarray}
U=\left(
    \begin{array}{c}
      U_{1} \\
      U_{2} \\
    \end{array}
  \right)
\sim\mathcal{N}_{2k}\left((\bm{0},\bm{\gamma}')',\left(
                                                   \begin{array}{cc}
                                                     \bm{A} & \bm{\Sigma} \\
                                                     \bm{\Sigma}' & \bm{\Phi} \\
                                                   \end{array}
                                                 \right)\label{distrU}
\right).
\end{eqnarray}
From the Cholesky decomposition, let $\bm{P}$ be a nonsingular matrix such
that $\bm{\Xi}=\bm{P}\bm{P}'$, and let
\begin{eqnarray}
Z=\bm{P}^{-1}(U_{1}-U_{2}) \quad{ } \mbox{ and } \quad{ }
\bm{R}=\bm{P}'\bm{P}. \label{zandr}
\end{eqnarray}
Further, let
\begin{eqnarray}
W=\left(
    \begin{array}{c}
      U_{1} \\
      Z \\
    \end{array}
  \right), \quad{ }
 \bm{F}=\left(
    \begin{array}{cc}
      \bm{0} & \bm{0} \\
      \bm{P} & \bm{0} \\
    \end{array}
  \right), \quad{ } \bm{B}=\left(
    \begin{array}{cc}
      \bm{0} & \bm{P}' \\
      \bm{P} & \bm{0} \\
    \end{array}
  \right),\label{BF}
\end{eqnarray}
\begin{eqnarray}
\eta=\textrm{E}\left[U'_{1}\bm{P}Z\,/Z'\bm{R}Z\right], \quad{ } \eta^{\ddag}=\textrm{E}\left[|U'_{1}\bm{P}Z|\,/Z'\bm{R}Z\right], \quad{ }
\omega=\textrm{E}\left[1\,/Z'\bm{R}Z\right]. \label{etaomega}
\end{eqnarray}
\begin{prn} \label{proetaomega}Suppose that Assumption~$(\mathcal{H}_{1})$ holds, then there exists $q_{\scriptscriptstyle 0}>0$ such that
\begin{enumerate}
  \item $|\eta(h)|\leqslant q_{\scriptscriptstyle 0}\eta^{\ddag}$ where $\eta^{\ddag}$  is defined in~\eqref{etaomega};
  \item $|\omega(h)|\leqslant q^{2}_{\scriptscriptstyle 0}\omega$ where $\omega$ is defined in~\eqref{etaomega}.
\end{enumerate}
\end{prn}
The proof of this proposition is given in the Appendix. From
Proposition~\ref{proetaomega}, it is clear that, in order to prove
that $\eta(h)$ and $\omega(h)$ are defined, it is sufficient to
prove that $\eta^{\ddag}<\infty$ and $\omega<\infty$. Below, we
establish a theorem which proves that, provided that $k\geqslant 3$,
$\omega<\infty$, and this implies that $\eta^{\ddag}<\infty$. To
introduce some notations, let $\mathbb{I}_{G}$ denote the indicator
function of the event $G$, let $H$ be $k\times k$-symmetric matrix,
let $\lambda(\bm{H})$ denote the eigenvalue of $\bm{H}$, and let
$\lambda_{1}(\bm{H})$, $\lambda_{2}(\bm{H})$, $\dots$,
$\lambda_{k}(\bm{H})$ be respectively the first, the second,
$\dots$, the $\mbox{k}^{th}$ the eigenvalue of $\bm{H}$.
\begin{prn}\label{born1}
Let $\alpha>0$, let
$\psi_{1}=\max\{|\lambda_{1}(\bm{B})|,|\lambda_{2}(\bm{B})|,\dots,|\lambda_{2k}(\bm{B})|\}$,
where $\bm{B}$ is defined in~\eqref{BF}, and let $W$ be the random vector
in~\eqref{BF}.  We have
\begin{eqnarray}
\textrm{\em E}\left\{\left(|U_{1}'\bm{P}Z|\,\big /Z'\bm{R}Z\right)\,
\mathbb{I}_{\left\{\|W\|\leqslant\alpha\right\}}\right\} \leqslant
\alpha^{2}\, \psi_{1}\,\omega\,/2.
\end{eqnarray}

\end{prn}
The proof of this proposition is given in the Appendix.
\begin{prn}\label{born2}
Let $\alpha>0$, and let $\psi_{0}=
\min\{\lambda_{1}(\bm{R}),\lambda_{2}(\bm{R}),\dots,\lambda_{k}(\bm{R})\}$.
We have
\begin{eqnarray}
\textrm{\em E}\left\{\left(|W'\bm{F}W|\,\big /\,Z'\bm{R}Z\right)
\mathbb{I}_{\left\{\|W\|>\alpha\right\}}\right\}\leqslant \psi_{1}
\,\left[\textrm{\em trace}(\bm{A})+k+\mu'\mu\right]\,\big
/(\alpha^{2}\,\psi_{0}).
\end{eqnarray}
\end{prn}
The proof of this proposition is given in the Appendix. By combining
Propositions~\ref{born1} and \ref{born2}, we establish the following
theorem which plays a central role in proving that $\eta$ exists
whenever $0<\omega<+\infty$.

\begin{thm}\label{mainthm}
Let $\alpha>0$, then there exists $M(\alpha)>0$ such that
\begin{eqnarray}
|\eta|\leqslant\eta^{\ddag}<M(\alpha)\omega+\,\psi_{1}\, \left(\textrm{\em
tr}\left(\bm{A}\right)+k+\mu'\mu\right)\,\big/(\alpha^{2} \psi_{0}),
\end{eqnarray}
where  $\psi_{1}$ and $\psi_{0}$ are given in
Propositions~\ref{born1} and \ref{born2} respectively.
\end{thm}
\begin{proof} Since $\bm{R}$ is a positive definite matrix,
we have
\begin{eqnarray}
|\eta|\leqslant \textrm{E}\left[|U'_{1}\bm{P}Z|\,\big
/(Z'\bm{R}Z)\right]=
\textrm{E}\left[|\left(Z',U'_{1}\right)\bm{F}\left(Z',U'_{1}\right)'|\,\big
/(Z'\bm{R}Z)\right],
\end{eqnarray}
where $\bm{F}$ is given in equation~\eqref{BF}.
Further, set $W=\left(Z',U'_{1}\right)'$.
We have,
\begin{eqnarray}
\textrm{E}\left[\frac{|W'\bm{F}W|}{Z'\bm{R}Z}\right]
=\textrm{E}\left\{\frac{|W'\bm{F}W|}{Z'\bm{R}Z}
\mathbb{I}_{\left\{\|W\|>\alpha\right\}}\right\}+\textrm{E}\left\{\frac{|W'\bm{F}W|}{Z'\bm{R}Z}
\mathbb{I}_{\left\{\|W\|\leqslant\alpha\right\}}\right\}.
\end{eqnarray}
Then, by combining Propositions~\ref{born1} and \ref{born2} and by
taking $M(\alpha)=\alpha^{2}\psi_{1}/2$, we get the stated result.
\end{proof}
By using Theorem~\ref{mainthm}, we establish the following corollary
which shows that $\eta^{\ddag}$ (and so $|\eta|$) is bounded by a positive real number which
is finite provided that $0<\omega<+\infty$.
\begin{cor}\label{corinterm}
Suppose that the conditions of Theorem~\ref{mainthm} hold. Then,
\begin{eqnarray*}
|\eta|\leqslant\eta^{\ddag}<\omega+\,\psi_{1}^{2}\, \left(\textrm{\em trace
}\left(\bm{A}\right)+k+\mu'\mu\right)\,\big/(2 \psi_{0}),
\end{eqnarray*}
where $\psi_{1}$ and $\psi_{0}$ are given in
Propositions~\ref{born1} and \ref{born2} respectively.
\end{cor}
\begin{proof}
For $\alpha=\sqrt{2/\psi_{1}}$, we have
$M(\alpha)=\alpha^{2}\psi_{1}/2=1$. Then, by using
Theorem~\ref{mainthm}, we get the statement of the corollary.
\end{proof}
\begin{rem}\label{rem:prop}
It should be noticed that Propositions~\ref{born1}-\ref{born2},
Theorem~\ref{mainthm}, and Corollary~\ref{corinterm} hold even in the
case where $W$ is not Gaussian, provided that the mean and the
variance-covariance matrix of $U$ are the same as the one given
in~\eqref{distrU}.
\end{rem}
\begin{rem}\label{rem:main}
From Corollary~\ref{corinterm}, it should be noticed that, if
$0<\omega<\infty$, then $|\eta|\leqslant\eta^{\ddag}<\infty$. This is an interesting finding
which shows that the nonzero correlation does not affect the
condition for the risk dominance of the SR estimator relative to the base
estimator. Thus, under normality, in order
to guarantee the existence of the \textrm{MSE} in~\eqref{mse} and
\eqref{mseopt}, it is sufficient to let $k\geqslant 3$.
\end{rem}
\begin{cor}\label{cor_main}
Under normality, $k\geqslant3$ implies $0<\omega<+\infty$ and
$|\eta|\leqslant\eta^{\ddag}<+\infty$.
\end{cor}
\begin{proof}
If $k\geqslant 3$, from the proof in JM~(2004), $0<\omega<+\infty$.
Then, by using Theorem~\ref{mainthm}, we get $|\eta|\leqslant\eta^{\ddag}<+\infty$.
\end{proof}
\begin{cor}\label{cor_mainh}
Suppose that Assumption~$(\mathcal{H}_{1})$ holds. Under normality, $k\geqslant3$ implies $0<\omega(h)<+\infty$ and
$|\eta(h)|<+\infty$.
\end{cor}
The proof follows from Proposition~\ref{proetaomega} and
Corollary~\ref{cor_main}.

Note that Corollary~\ref{cor_main} generalizes the main theorem
in~JM~(2004). Further, Corollary~\ref{cor_mainh} extends the result
of Corollary~\ref{cor_main} to a class of SR-type estimators which
includes the SR-type estimator in JM~(2004) as a special case.
\subsection{Extension to elliptically contoured random samples }
In this subsection, we show that the result given in
Corollary~\ref{cor_main} remains valid in the context of some
elliptically contoured random samples. The importance of such a
family of distributions is the primary source of our motivation.
Indeed, as discussed in the literature, elliptically contoured
distributions have been particularly useful in several areas of
applications such as actuarial science~(see Furman
and Landsman,~2006, Landsman and Valdez,~~2003), or
economics and finance (see Bingham and Kiesel,~2001).

Recall that  a class of elliptically contoured distributions
includes for example the multivariate Gaussian, t, Pearson type II and
VII, as well as Kotz distributions. To simplify the notation, let
$X\sim\mathcal{E}_{q}\left(\mu,\,\bm{\Sigma};g\right)$ stand for a
$q$-column random vector distributed as an elliptically contoured
vector with mean $\mu$ and scale parameter matrix
$\bm{\Sigma}$, where $\bm{\Sigma}$ is a positive definite matrix, and
$g$ is the
probability density function (p.d.f) generator.
For the sake of simplicity, we consider the case where the p.d.f of
$X\sim \mathcal{E}_{q}\left(\mu,\,\bm{\Sigma};g\right)$ is assumed
to be written as
\begin{eqnarray}
f_{X}(x) =\ds{\int_{0}^{\infty}}f_{\mathcal{N}\left(\mu,\,
z^{-1}\bm{\Sigma}\right)}(x)\kappa(z)dz, \label{mixture}
\end{eqnarray}
where $f_{\mathcal{N}\left(\mu,\, \bm{\Sigma}\right)}$ denotes the
p.d.f of a random vector which follows a normal distribution with
mean $\mu$ and variance-covariance $\bm{\Sigma}$, and $\kappa(.)$ is
a weighting function that satisfies
$\ds{\int_{0}^{\infty}}\ds{\frac{1}{t}}\left|\kappa(t)\right|\,dt<\infty$.
Note that the weighting function $\kappa(.)$ does not need to be
nonnegative. In the case where the function $\kappa(.)$ is
nonnegative, then $\kappa(.)$ is a p.d.f, and the subclass of
elliptically contoured distributions is known as a mixture of
multivariate normal distributions. For more details, we refer to
Chmielewski~(1981), Gupta and Verga~(1995), Nkurunziza and
Chen~(2013) among
others. In particular, Gupta and Verga~(1995) give the conditions on the p.d.f generator $g$
for the pdf of $X\sim
\mathcal{E}_{q}\left(\mu,\,\bm{\Sigma};g\right)$ to be rewritten as
in~\eqref{mixture}. From now on, we suppose that
$\left((\hat{\beta}-\beta)',\,(\tilde{\beta}-\beta)'\right)'$ has a
p.d.f which can be rewritten as in ~\eqref{mixture}, and in a
similar way to Section~\ref{sec:classSR}, let
\begin{eqnarray}
U=\left(
    \begin{array}{c}
      U_{1} \\
      U_{2} \\
    \end{array}
  \right)=\left(
            \begin{array}{c}
              \hat{\beta}-\beta \\
              \tilde{\beta}-\beta \\
            \end{array}
          \right)
\sim\mathcal{E}_{2k}\left((0,\gamma')',\left(
                                                   \begin{array}{cc}
                                                     \bm{A} & \bm{\Sigma} \\
                                                     \bm{\Sigma}' & \bm{\Phi} \\
                                                   \end{array}
                                                 \right)\label{ellipticallyU}
;g\right),
\end{eqnarray}
where $\gamma$, $\bm{A}$, $\bm{\Sigma}$, $\bm{\Phi}$ are as defined
in Section~\ref{sec:classSR}.
\begin{thm}\label{propmain}
Suppose that
$\left((\hat{\beta}-\beta)',\,(\tilde{\beta}-\beta)'\right)'$ is
distributed as in \eqref{ellipticallyU} and suppose that the weighting function $\kappa(.)$ satisfies
\quad{ }
$0<\ds{\int_{0}^{\infty}}t\left|\kappa(t)\right|\,dt<\infty$. Then,
\mbox{ }  $k\geqslant3$ \mbox{ } implies \mbox{ } $0<\omega<+\infty$
and $|\eta|\leqslant\eta^{\ddag}<+\infty$.
\end{thm}
\begin{proof}
First, recall that the family of elliptically contoured
distribution is closed under linear transformations. Then, if
$\left((\hat{\beta}-\beta)',\,(\tilde{\beta}-\beta)'\right)'$ is
distributed as in \eqref{ellipticallyU},~\eqref{mse} and \eqref{mseopt} hold, and
then~\eqref{etaomega} holds with $Z\sim
\mathcal{E}_{k}\left(\mu,\bm{I}_{k};g\right)$. Therefore, by using
Remark~\ref{rem:prop}, we conclude that Corollary~\eqref{corinterm}
holds. Further, as in JM~(2004), we get
\begin{eqnarray*}
\ds{\frac{1}{\max\left(\lambda_{1}(\bm{R}),\lambda_{2}(\bm{R}),\dots,\lambda_{k}(\bm{R})\right)}}
\textrm{E}\left(\frac{1}{Z'Z}\right)<\omega<\ds{\frac{1}{\min\left(\lambda_{1}(\bm{R}),
\lambda_{2}(\bm{R}),\dots,\lambda_{k}(\bm{R})\right)}}
\textrm{E}\left(\frac{1}{Z'Z}\right).\label{courantthm}
\end{eqnarray*}
Then, it suffices to prove that $\textrm{E}\left(1\,\Big
/Z'Z\right)<\infty$ for all $k\geqslant3$.

From~\eqref{mixture} and Fubini's Theorem, we
have
\begin{eqnarray*}
\textrm{E}\left(1\,\Big /Z'Z\right)=\int_{0}^{\infty}\kappa(t)
\textrm{E}_{t,\mu}\left(1\, \Big /U_{0}'U_{0}\right) dt,
\end{eqnarray*}
where $U_{0}\sim
\mathcal{N}_{k}\left(\mu,\,t^{-1}\bm{I}_{k}\right)$. Note that,
\begin{eqnarray*}
\textrm{E}_{t,\mu}\left(\frac{1}{U_{0}'U_{0}}\right)=
t\textrm{E}_{t,\mu}\left(\frac{1}{(\sqrt{t}U_{0})'(\sqrt{t}U_{0})}\right)
=t\textrm{E}\left(\frac{1}{\chi^{2}_{k}(t\mu'\mu)}\right),
\end{eqnarray*}
and then,
\begin{eqnarray*}
0<\textrm{E}\left(1\, \Big /Z'Z\right)=\int_{0}^{\infty}t\kappa(t)
\textrm{E}\left(\chi^{-2}_{k}(t\mu'\mu)\right)
dt<\frac{1}{k-2}\int_{0}^{\infty}t\left|\kappa(t)\right| dt<+\infty,
\,
\end{eqnarray*}
\mbox{for all $k\geqslant 3$}. This completes the proof.
\end{proof}
\begin{cor}\label{coretah}
Suppose that Assumption~$(\mathcal{H}_{1})$ holds. Also, suppose that
$\left((\hat{\beta}-\beta)',\,(\tilde{\beta}-\beta)'\right)'$ is
distributed as in \eqref{ellipticallyU} and suppose that the weighting function $\kappa(.)$ satisfies 
\quad{ }\\
$0<\ds{\int_{0}^{\infty}}t\left|\kappa(t)\right|\,dt<\infty$. Then,
\mbox{ }  $k\geqslant3$ \mbox{ } implies \mbox{ }
$0<\omega(h)<+\infty$ and $|\eta(h)|<+\infty$.
\end{cor}
The proof follows by combining Proposition~\ref{proetaomega} and Theorem~\ref{propmain}.
\begin{rem}
Note that in the Gaussian case, the weighting function $\kappa(t)$
is the Dirac delta function at $t-1$~(see Gupta and Varga,~1975). Thus, the conditions of
Theorem~\ref{propmain} hold since
$\ds{\int_{0}^{\infty}}t\left|\kappa(t)\right|dt=1$. This shows that
Corollary~\ref{cor_main} and Corollary~\ref{cor_mainh} are  special
cases of Theorem~\ref{propmain} and Corollary~\ref{coretah}
respectively.
\end{rem}
\subsection{Further extensions and statistical practice}\label{sec:extension}
\subsubsection{Singular distributions case}
In the previous sections, we derived the results under the
assumption that the joint distribution of $\hat{\beta}$ and
$\tilde{\beta}$ is not singular (see the
relation~\eqref{distrbeta}). This is a limitation which excludes,
for example, the case where the imprecise prior information is in
the form of a linear restriction between the parameters.
Nevertheless, this is particulary the case where there is a
restriction binding some regression coefficients. Indeed, such a
situation is common in economic theory where for example, as
introduced by Douglas and Cobb~(1928), the sum of the exponents in a
Cobb-Douglas production is known to  be one. Thus, in this
subsection, we consider that
$\left((\hat{\beta}-\beta)',\,(\tilde{\beta}-\beta)'\right)$ has the
same distribution as in~\eqref{distrbeta} where the matrices
$\bm{\Xi}$ and $\bm{\Phi}$ are (possibly) singular. For this kind of
problem, the joint distribution of $\hat{\beta}$ and $\tilde{\beta}$
is (possibly) singular and thus, it is important to show how the
proposed methodology works in this case. To this end, let $q$ be the
rank of $\bm{\Xi}$ with $q\leqslant k$. Briefly,  we show that,
under some conditions, the established results hold by replacing $k$
by $q$. Namely, a sufficient condition for the risk dominance of any
member of the class of SR-type estimators relative to the base
estimator is to let $q\geqslant3$. Of course, this condition implies
that $k\geqslant3$ since $k\geqslant q$. Namely, we suppose that the
following conditions hold.

\noindent {\bf \underline{Assumption $(\mathcal{H}_{2})$}}: {\em The
function $h\left(\hat{\beta},\tilde{\beta}\right)$ is a measurable
function of $\hat{\beta}-\tilde{\beta}$ only}.

\begin{rem}
Note that the function $h$ which gives the SR estimator satisfies
Assumption~$(\mathcal{H}_{2})$. Namely, for the SR estimator, we
have
$h\left(\hat{\beta},\tilde{\beta}\right)=\left\|\hat{\beta}-\tilde{\beta}\right\|^{-2}$.
\end{rem}

\noindent {\bf \underline{Assumption $(\mathcal{H}_{3})$}}: {\em There exists a symmetric and positive definite matrix $\bm{\Lambda}$ such that $\bm{\Lambda}^{1/2}\bm{\Xi}\bm{\Lambda}^{1/2}$ is idempotent and $\bm{\Lambda}\bm{\Xi}\bm{\Lambda}\gamma=\bm{\Lambda}\gamma$}.

\begin{rem}
It should be noted that in the case where $\bm{\Xi}$ is invertible,
it suffices to take $\bm{\Lambda}=\bm{\Xi}^{-1}$. Below we give
another, more specific, example of a matrix $\bm{\Lambda}$ in the
case where the prior information is a linear restriction on the
regression coefficients.
\end{rem}
\begin{thm} \label{furtherext} Suppose that Assumptions~$(\mathcal{H}_{1})$-$(\mathcal{H}_{3})$
hold. Under normality, \\$k\geqslant q\geqslant
3$ implies $0<\omega(h)<\infty$ and $|\eta(h)|<+\infty$.
\end{thm}
The proof of this theorem is given in the Appendix.
\subsubsection{Special singular case: Linear restriction}
In this subsection, we show that the proposed methodology works in a
very special case where the uncertain prior information refers to a
certain linear restriction. In particular, we consider the case
where the restriction is of the form
\begin{eqnarray}
\bm{R}\beta=r,\label{constraint}
\end{eqnarray}
where $\bm{R}$ is a known $q\times k$-full matrix with $q\leqslant
k$; $r$ is a known $q$-column vector. With a suitable choice of the
matrix $\bm{R}$ and the vector $r$, the constraint~\eqref{constraint} yields the
case where some regression coefficients are not statistically
significant i. e. their corresponding explanatory variables should
be excluded from the model.

Under the constraint in~\eqref{constraint}, the restricted estimator
 for $\beta$ is
$\tilde{\beta}=\hat{\beta}+\bm{J}\left(\bm{R}\hat{\beta}-r\right)$,
where
$\bm{J}=(\bm{X}'\bm{X})^{-1}\bm{R}'\left[\bm{R}(\bm{X}'\bm{X})^{-1}\bm{R}'\right]^{-1}$.
Then, if the restriction in~\eqref{constraint} does not hold and if
the error is normally distributed, it can be also verified that
\begin{eqnarray}
\left(
  \begin{array}{c}
    \hat{\beta}-\beta \\
    \tilde{\beta}-\beta \\
  \end{array}
\right)\sim
\mathcal{N}_{2k}\left(\left(\bm{0},\gamma'\right)',\left(
                                                     \begin{array}{cc}
                                                       \bm{A} & \bm{A}-\bm{J}\bm{R}\bm{A} \\
                                                       \bm{A}-\bm{J}\bm{R}\bm{A} & \bm{A}
                                                       -\bm{J}\bm{R}\bm{A} \\
                                                     \end{array}
                                                   \right)
\right),\label{distconstraint}
\end{eqnarray}
where $\gamma=\bm{J}\left(\bm{R}\beta-r\right)$ and $\bm{A}=\sigma^{2}\left(\bm{X}'\bm{X}\right)^{-1}$. Thus, here, the variance-covariance
matrix of $(\hat{\beta}',\tilde{\beta}')$ is singular and so is the
variance-covariance matrix of $\tilde{\beta}$. The following proposition
shows that the Assumption~$(\mathcal{H}_{3})$ holds
by taking $\bm{\Lambda}=\bm{A}^{-1}$.  Thus, the proposed
methodology works in this practical case.

\begin{prn} \label{furtherexts} Suppose that the base and restricted estimators
follow the distribution in~\eqref{distconstraint} and let
$\bm{\Lambda}=\bm{A}^{-1}$, then, Assumption~$(\mathcal{H}_{2})$
holds.
\end{prn}
\begin{proof}
We have $\bm{\Xi}=\bm{J}\bm{R}\bm{A}$ and
$\gamma=\bm{J}\left(\bm{R}\beta-r\right)$. Then, the proof follows
after applying standard algebraic computations.
\end{proof}
\begin{rem}
Actually, an even more general result could be proved. Indeed, by
using the similar transformation as in Nkurunziza~(2013), one can extend
Theorem~\ref{furtherext} to the case of singular elliptically
contoured distribution.
\end{rem}
\section{A simulation study and data analysis}\label{sec:simuldata}
\subsection{A Simulation Study }\label{sec:simulat} In this section, we
carry out  Monte Carlo  simulation studies to examine the mean
square error (MSE) performance of the SPSL over the base estimator.
To this end, we follow the similar sampling experiments as in
JM~(2004). Namely, for $k=3$ and $k=4$, we consider the general
linear model
\[
Y_{i}=X[i,.]\beta+\epsilon_{i}=\sum_{j=1}^{k}\beta_{j}X[i,j]+\epsilon_{i},
\quad{ } \mbox{ for } i=1,2,\dots,n,
\]
for small and large sample sizes. In order to save the space, we report only
 the results for $n=15$
and $n=25$. 
Although, not reported here, similar
results hold for $n=50$ and $n=125$ (they are available
from the author upon request). For the dimension of the parameter
vector $\beta$, here, we focus only on the
cases where $k=3$ and $k=4$, as the case $k=5$ has been
studied in JM~(2004). The $n\times k$-matrix $\bm{X}$ and the noise
$\epsilon$ were generated by following the sampling design described
in JM~(2004). For the convenience of the reader, we outline below
this sampling design.

Briefly, as in the quoted paper, for $k=3$ and $k=4$, the first
column of the $n\times k$ matrix $\bm{X}$ is a column of unit values
and the remaining columns of the $X[i,.]$'s are generated
independently from a $(k-1)$-dimensional normal distribution with a
mean vector of 1s, standard deviations all equal to 1, and various
levels of pairwise correlations. Further, the observations of
the $\epsilon_{i}$'s were generated independently based on various
normal probability distributions, all defined to have zero means
over a range of standard deviations. For every sample size, 5,000
replications were carried out in order to compute the empirical
quadratic risk estimates.

As in JM~(2004), we take
$\tilde{\beta}=\left[\textrm{diag}\left(\bm{X}'\bm{X}\right)\right]^{-1}\bm{X}'y$,
where $\textrm{diag}\left(\bm{X}'\bm{X}\right)$ denotes a $k\times
k$-diagonal matrix. Further, as in JM~(2004), the comparison between
the SPLS and LS estimators is based on the quantity called the
relative mean square efficiency (RMSE) of the estimators with
respect to LS, namely
$$
\textrm{RMSE}\left( \mbox{proposed
estimator}\right)=\mbox{risk}\left(\mbox{proposed
estimator}\right)\Big/
  \mbox{risk}
  \left(LS \right).
  $$
Therefore, we
  have
\begin{eqnarray*}
\textrm{RMSE}\left(LS\right)&=&
\mbox{risk}\left(LS\right)\Big/\mbox{risk}\left(LS\right)=1, \,
\textrm{RMSE}\left(SPLS\right)=\mbox{risk}\left(SPLS\right)\Big/\mbox{risk}\left(LS\right).
\end{eqnarray*}
  Thus,
    a relative efficiency
less than one indicates the degree of superiority of the new
estimator over the LS estimator.

\begin{figure}[htbp]
\centering
\subfigure[\mbox{$n=15$,  $\sigma=0.1$, $k=3$}]{
\label{T15p6} 
\includegraphics[height=1.65in,width=2.75in]{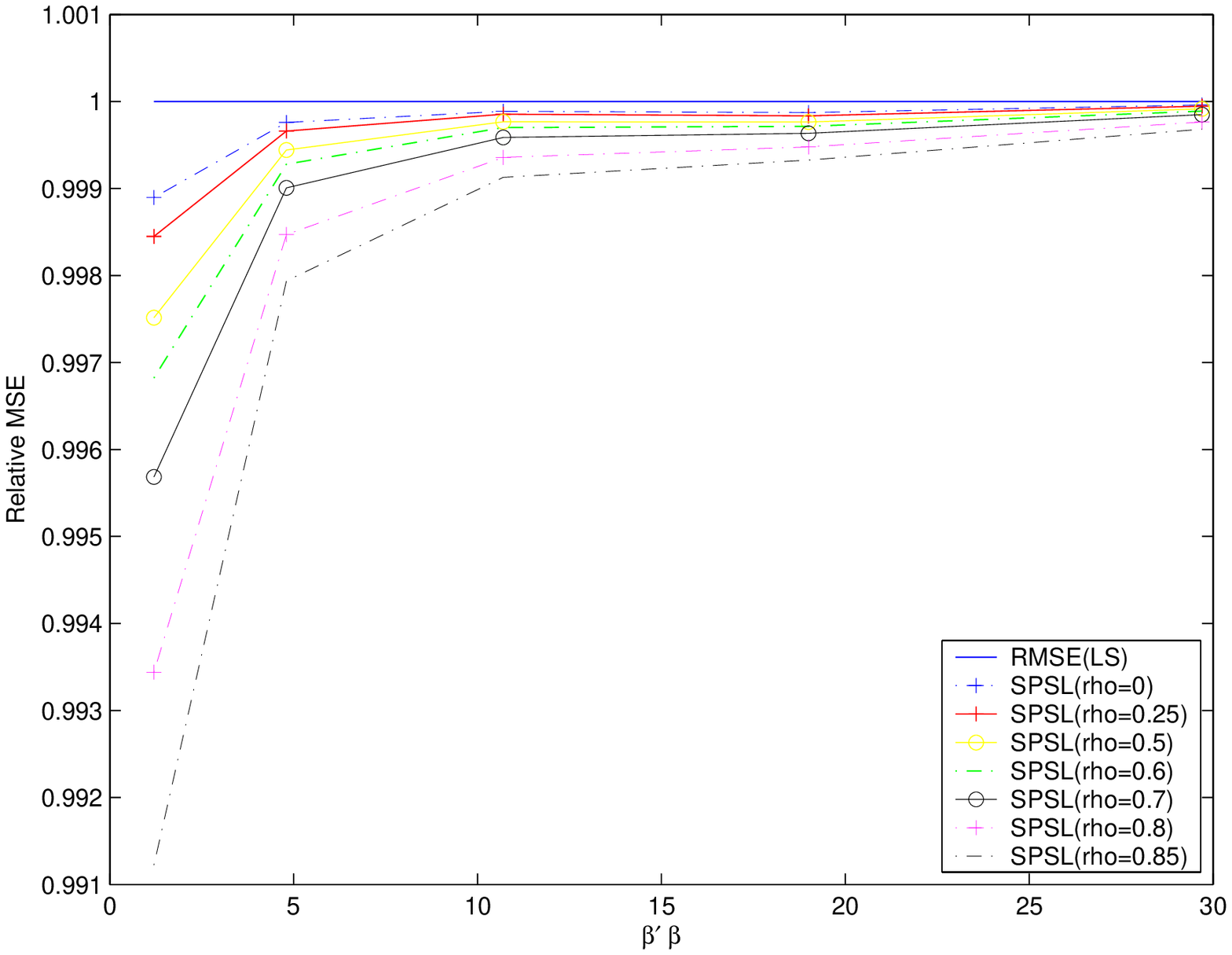}
} \subfigure[\mbox{$n=15$, $\sigma=0.25$, $k=3$} ]{
\label{T30p6} 
\includegraphics[height=1.65in,width=2.75in]{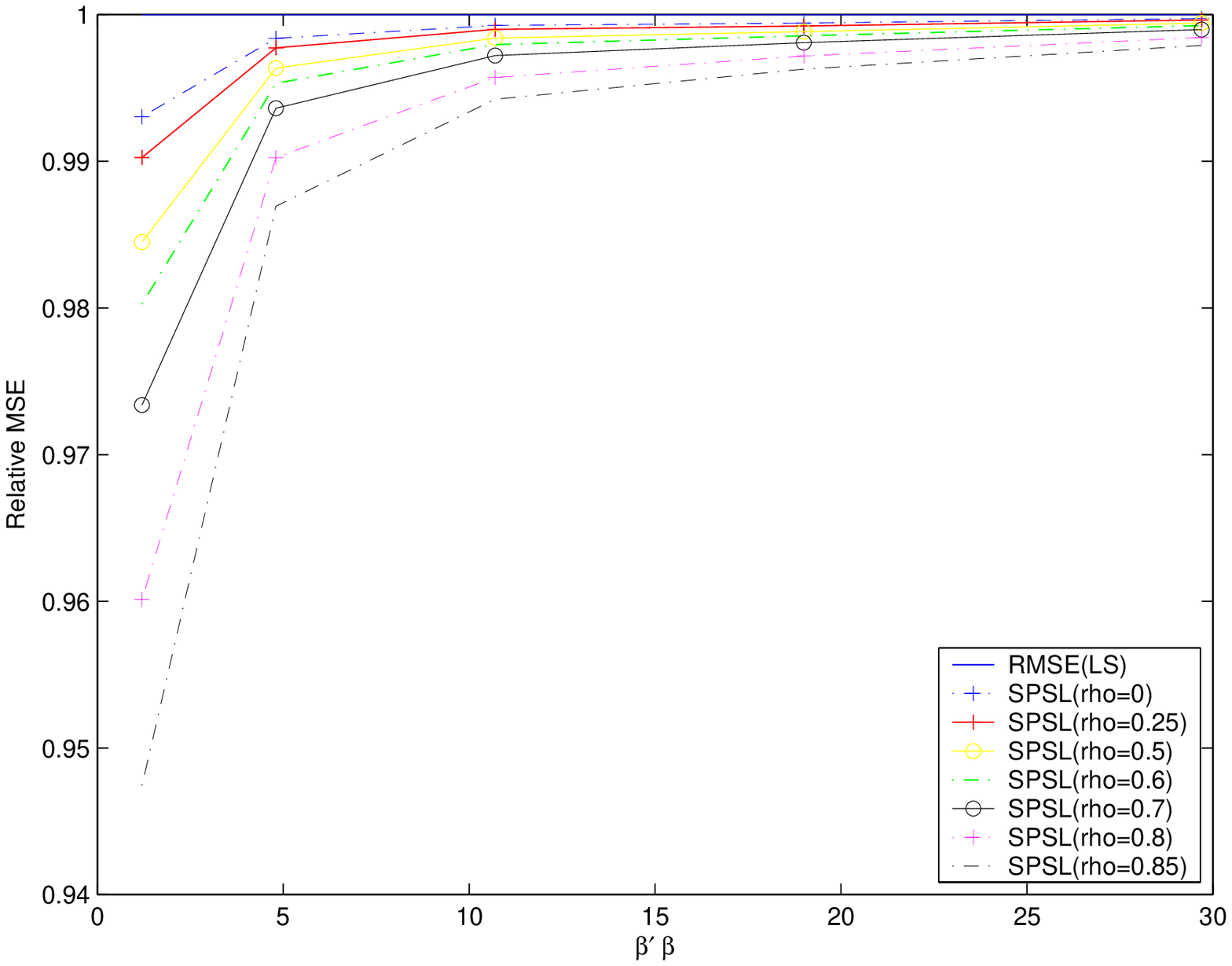}
} \subfigure[\mbox{$n=15$, $\sigma=0.5$, $k=3$}]{
\label{T40p6} 
\includegraphics[height=1.65in,width=2.75in]{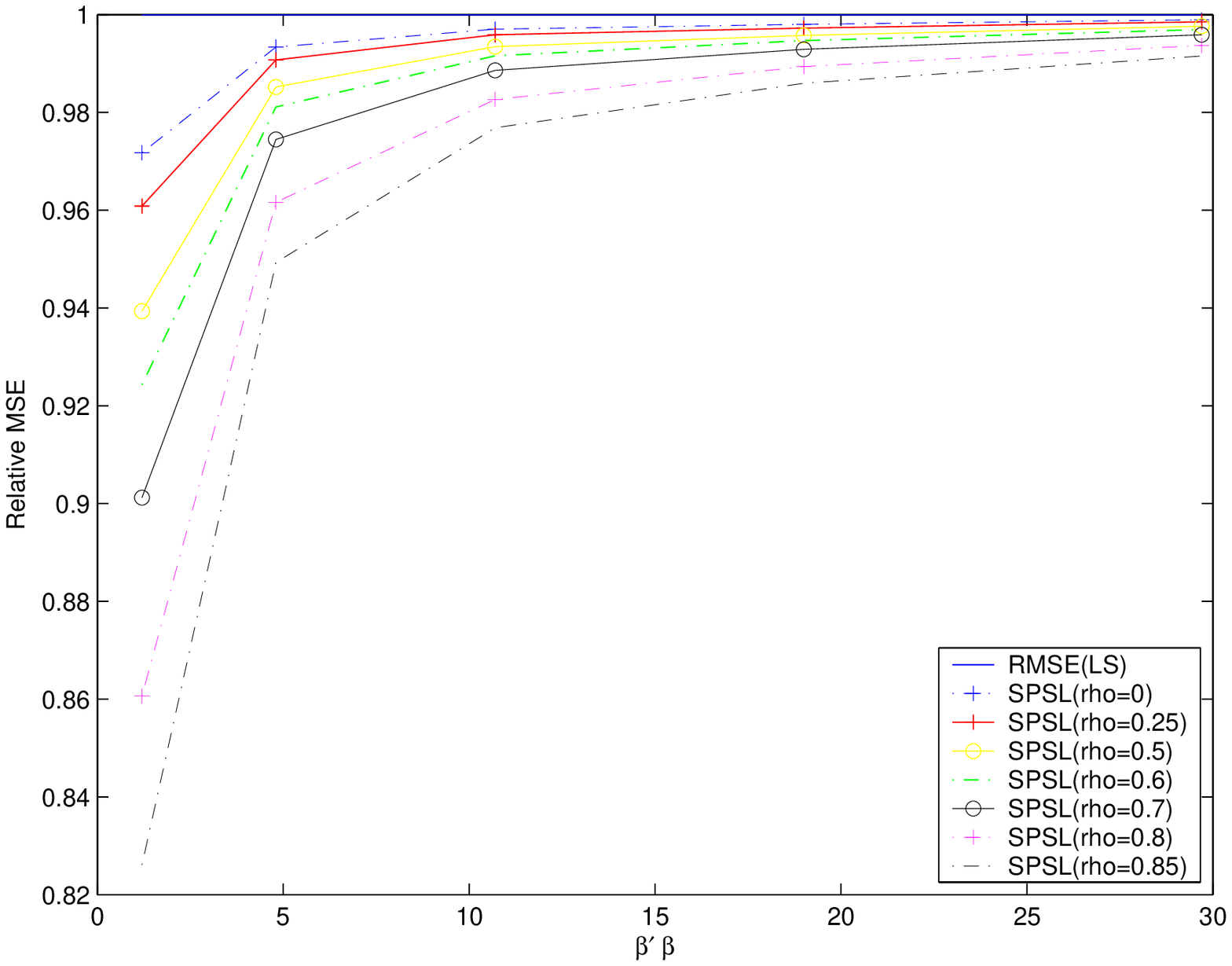}
} \subfigure[\mbox{$n=15$, $\sigma=1$, $k=3$}]{
\label{T50p6} 
\includegraphics[height=1.65in,width=2.75in]{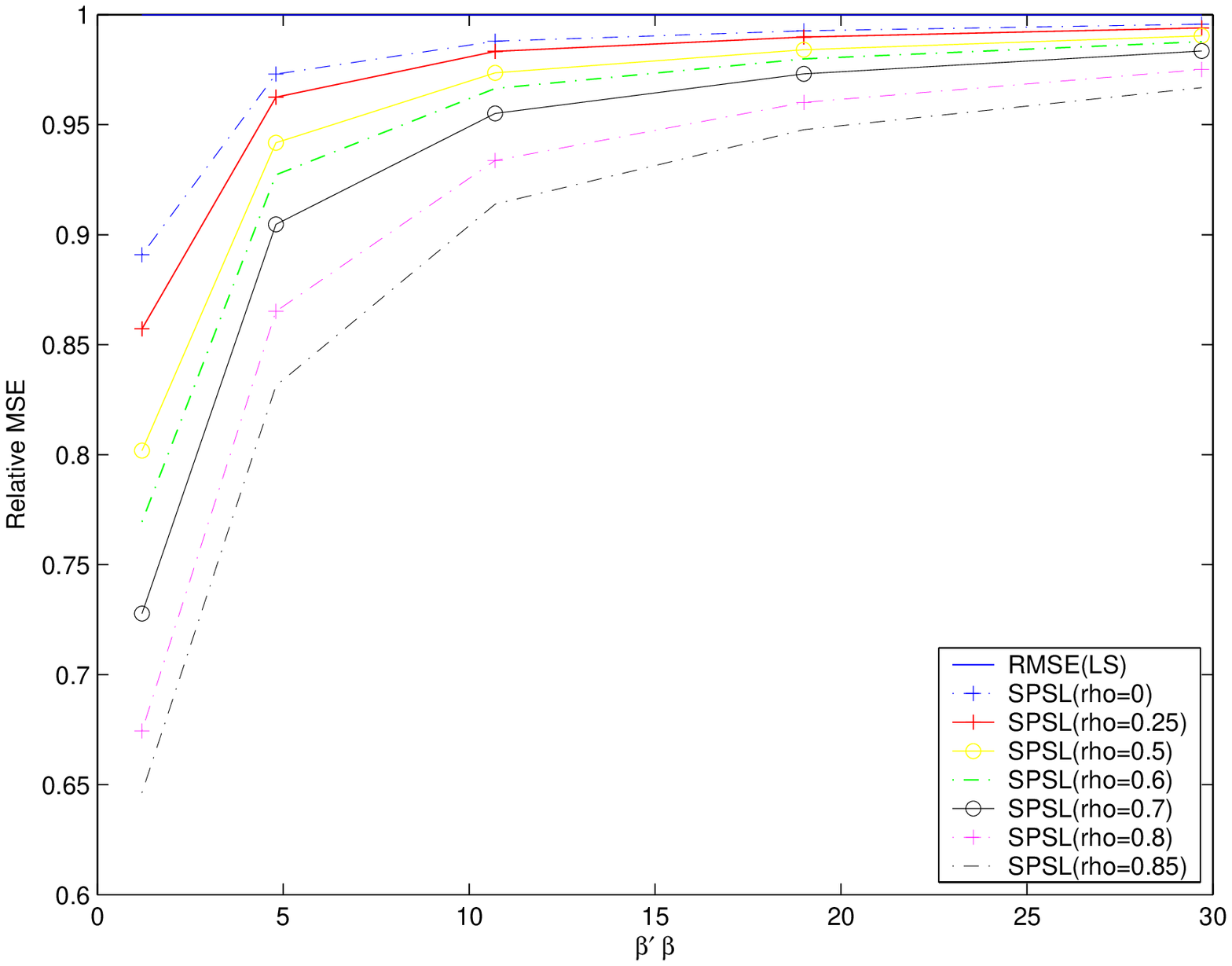}
} \subfigure[\mbox{$n=25$,  $\sigma=0.1$, $k=3$}]{
\label{T15p6} 
\includegraphics[height=1.65in,width=2.75in]{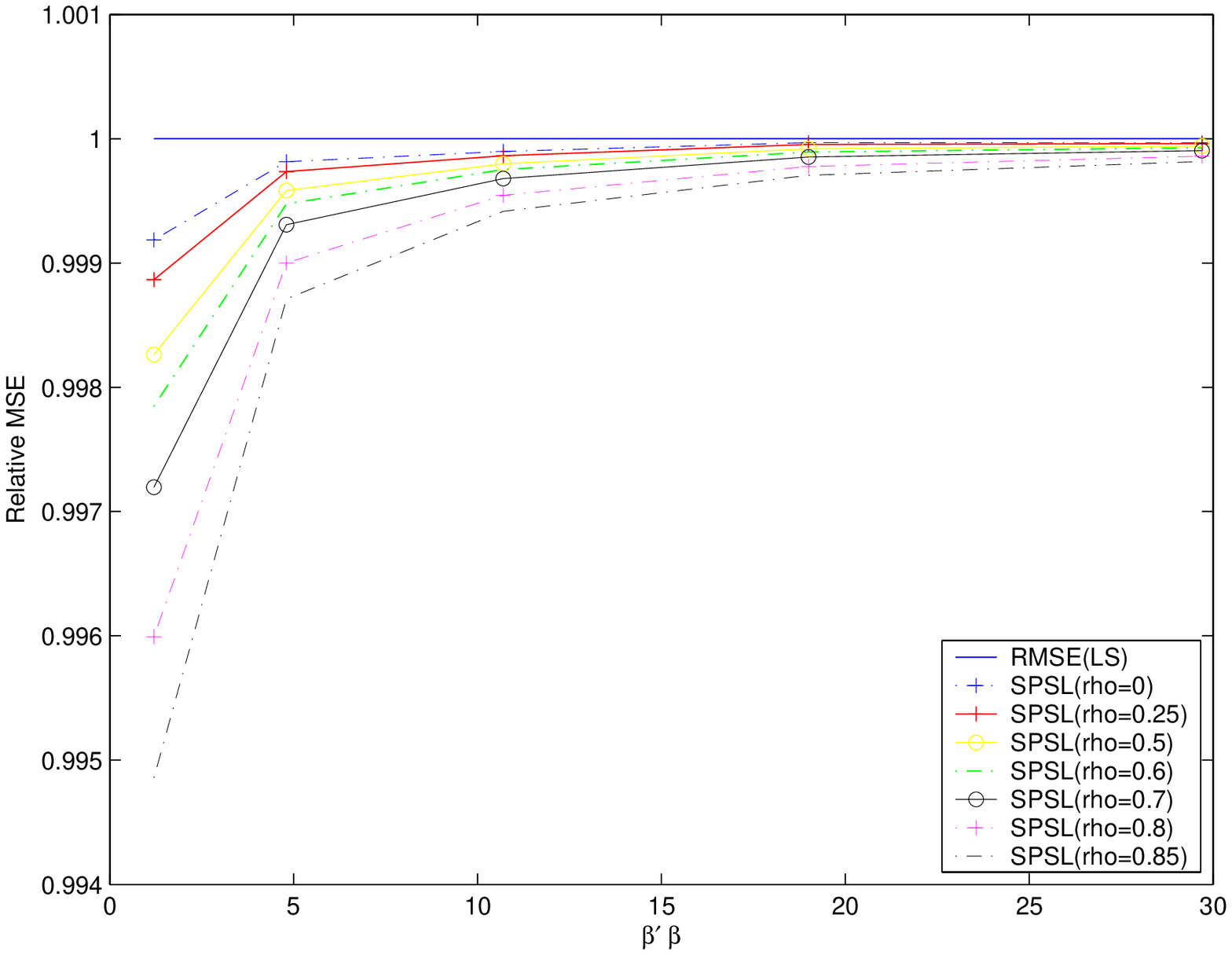}
} \subfigure[\mbox{$n=25$, $\sigma=0.25$, $k=3$} ]{
\label{T30p6} 
\includegraphics[height=1.65in,width=2.75in]{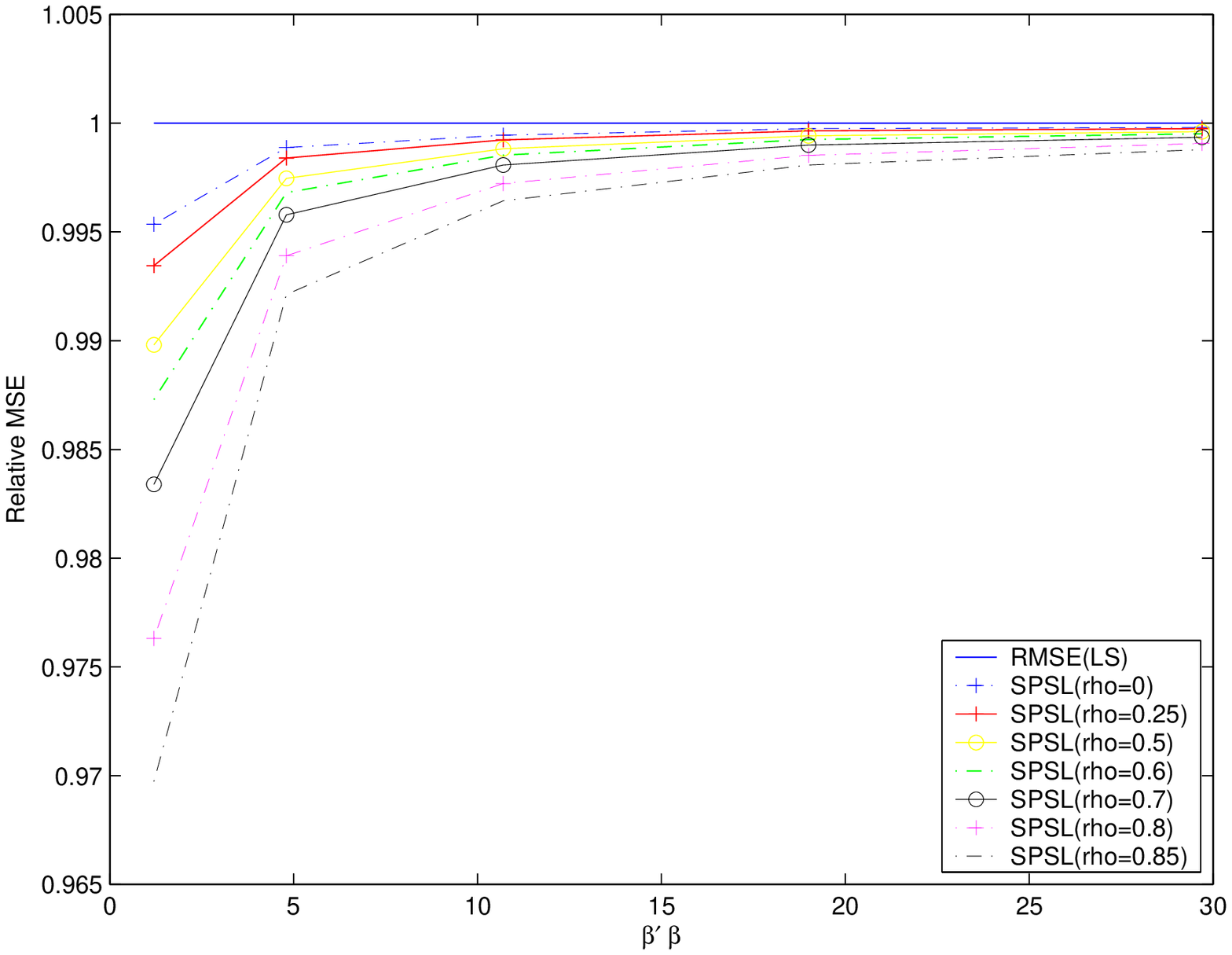}
} \subfigure[\mbox{$n=25$, $\sigma=0.5$, $k=3$}]{
\label{T40p6} 
\includegraphics[height=1.65in,width=2.75in]{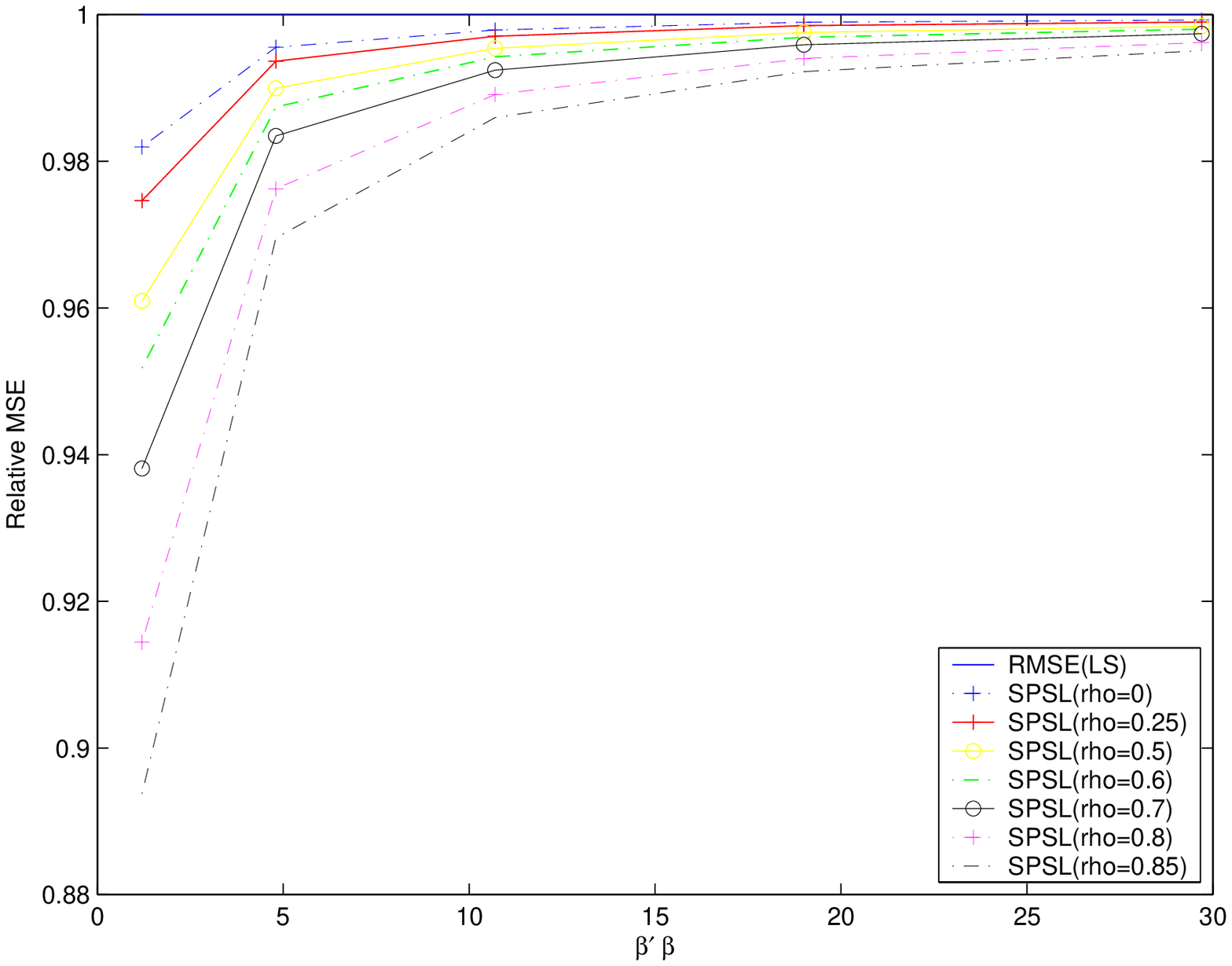}
} \subfigure[\mbox{$n=25$, $\sigma=1$, $k=3$}]{
\label{T50p6} 
\includegraphics[height=1.65in,width=2.75in]{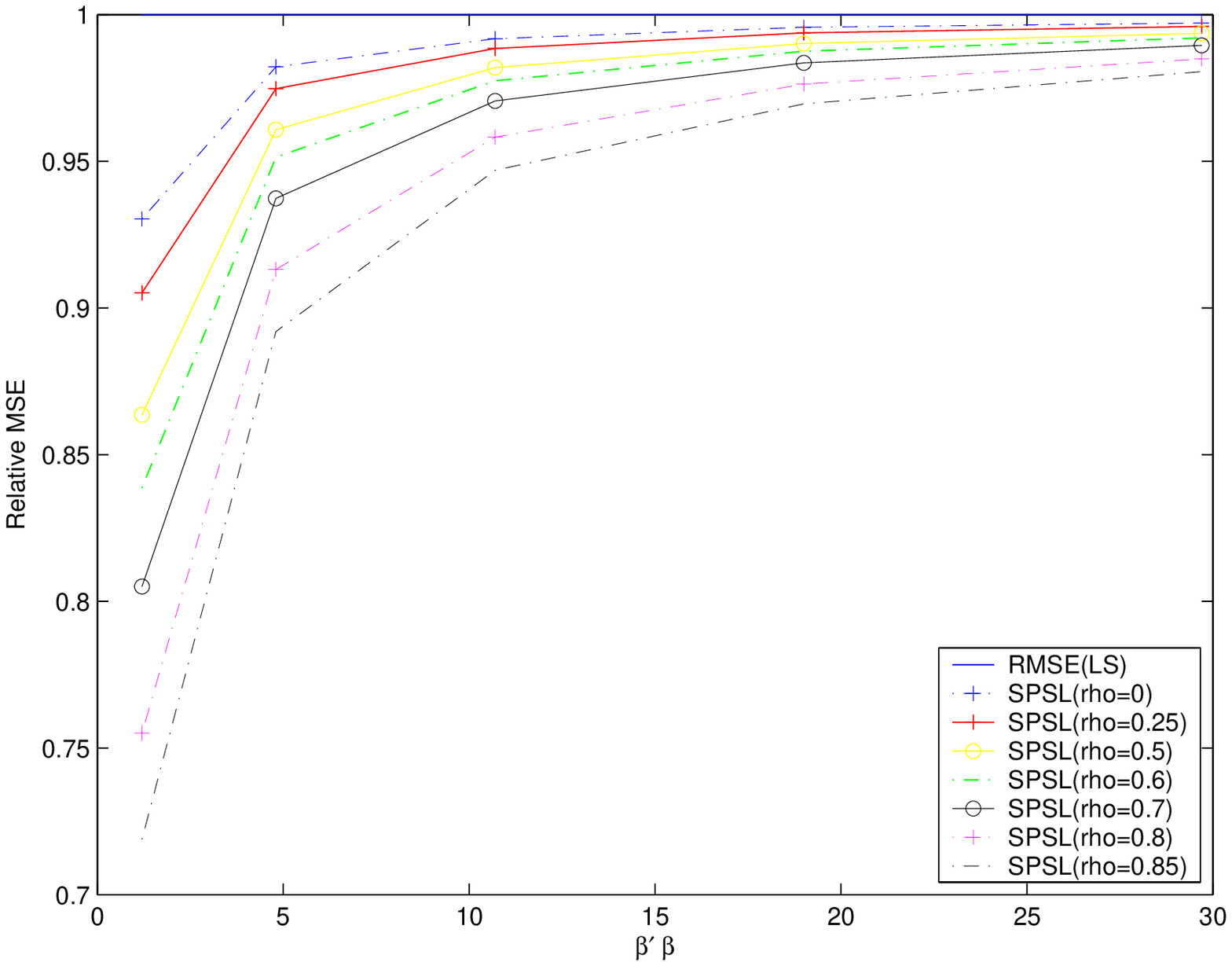}
}
 \caption{Relative efficiency versus \mbox{$\beta'\beta$ }}
\label{figloopbeta}
\end{figure}

\begin{figure}[htbp]
\centering
\subfigure[\mbox{$n=15$,  $\sigma=0.1$, $k=3$}]{
\label{T15p6} 
\includegraphics[height=1.65in,width=2.75in]{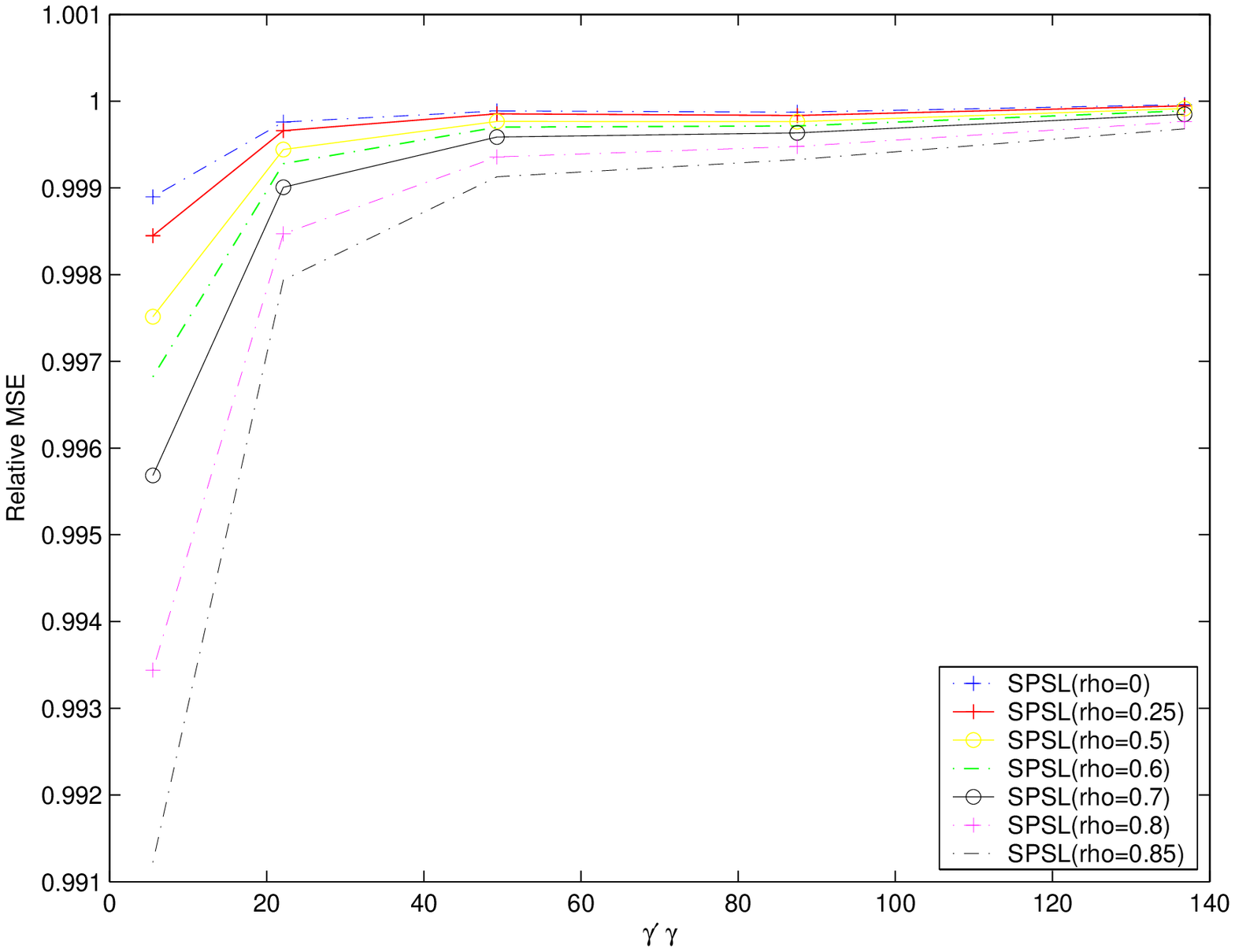}
} \subfigure[\mbox{$n=15$, $\sigma=0.25$, $k=3$} ]{
\label{T30p6} 
\includegraphics[height=1.65in,width=2.75in]{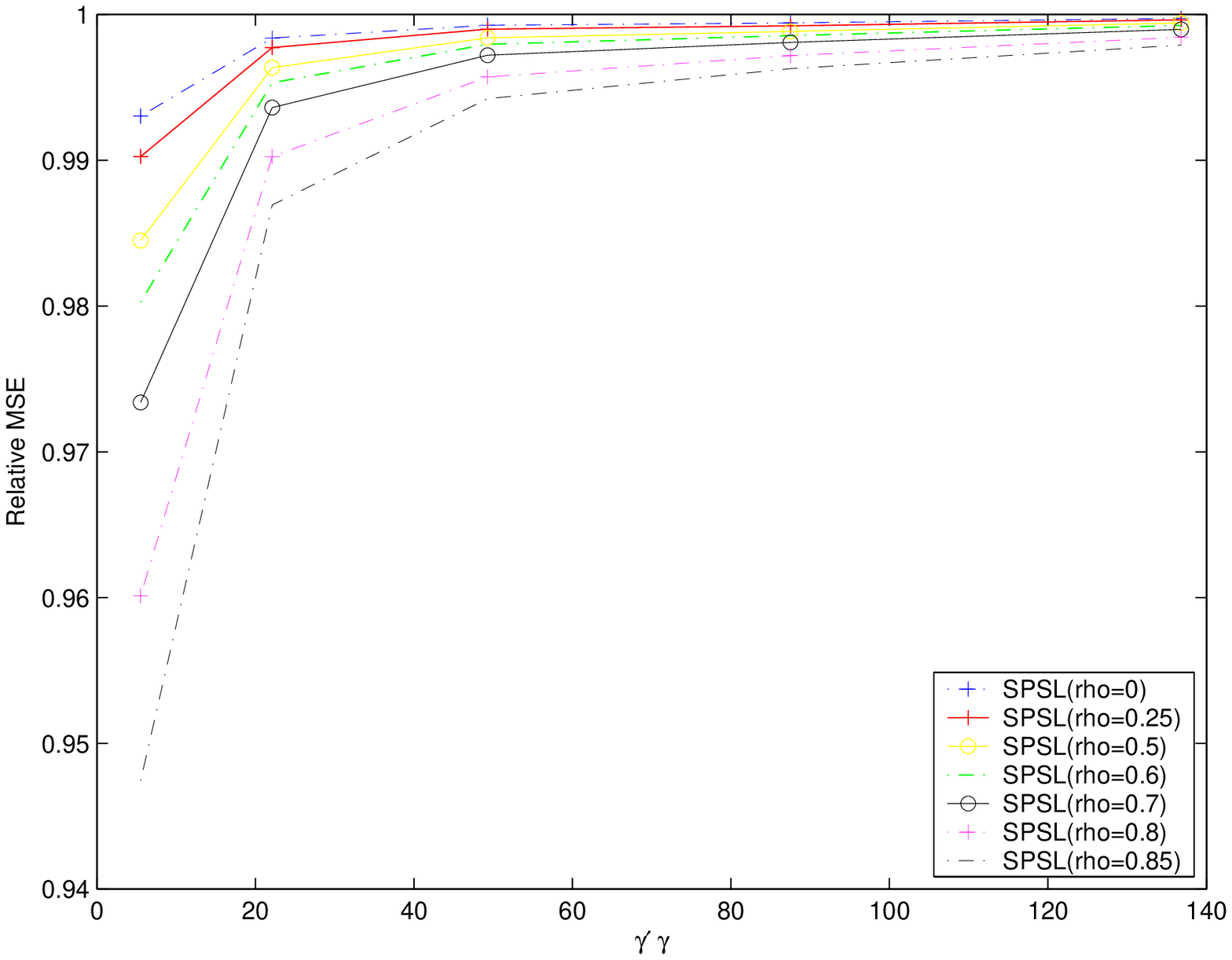}
} \subfigure[\mbox{$n=15$, $\sigma=0.5$, $k=3$}]{
\label{T40p6} 
\includegraphics[height=1.65in,width=2.75in]{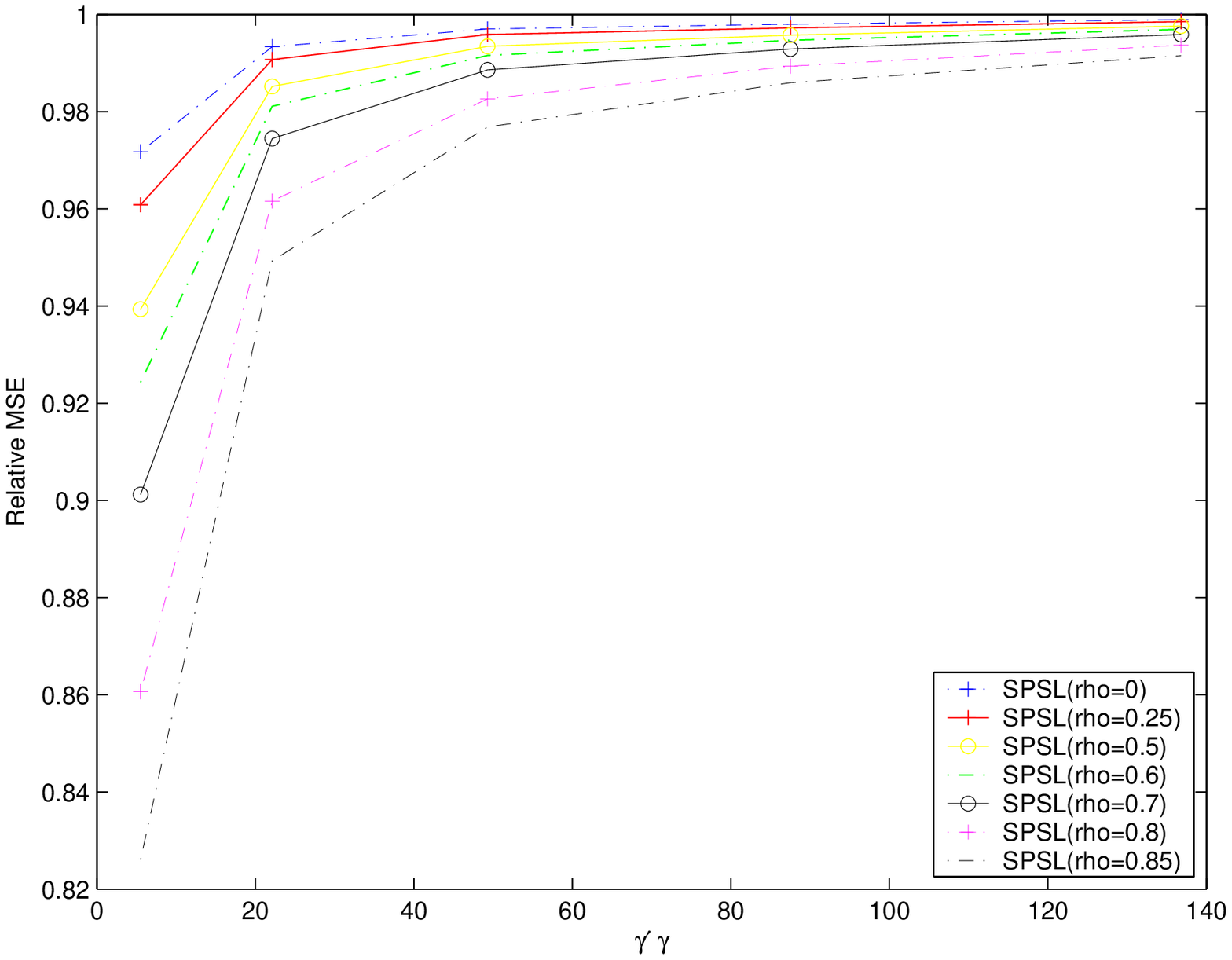}
} \subfigure[\mbox{$n=15$, $\sigma=1$, $k=3$}]{
\label{T50p6} 
\includegraphics[height=1.65in,width=2.75in]{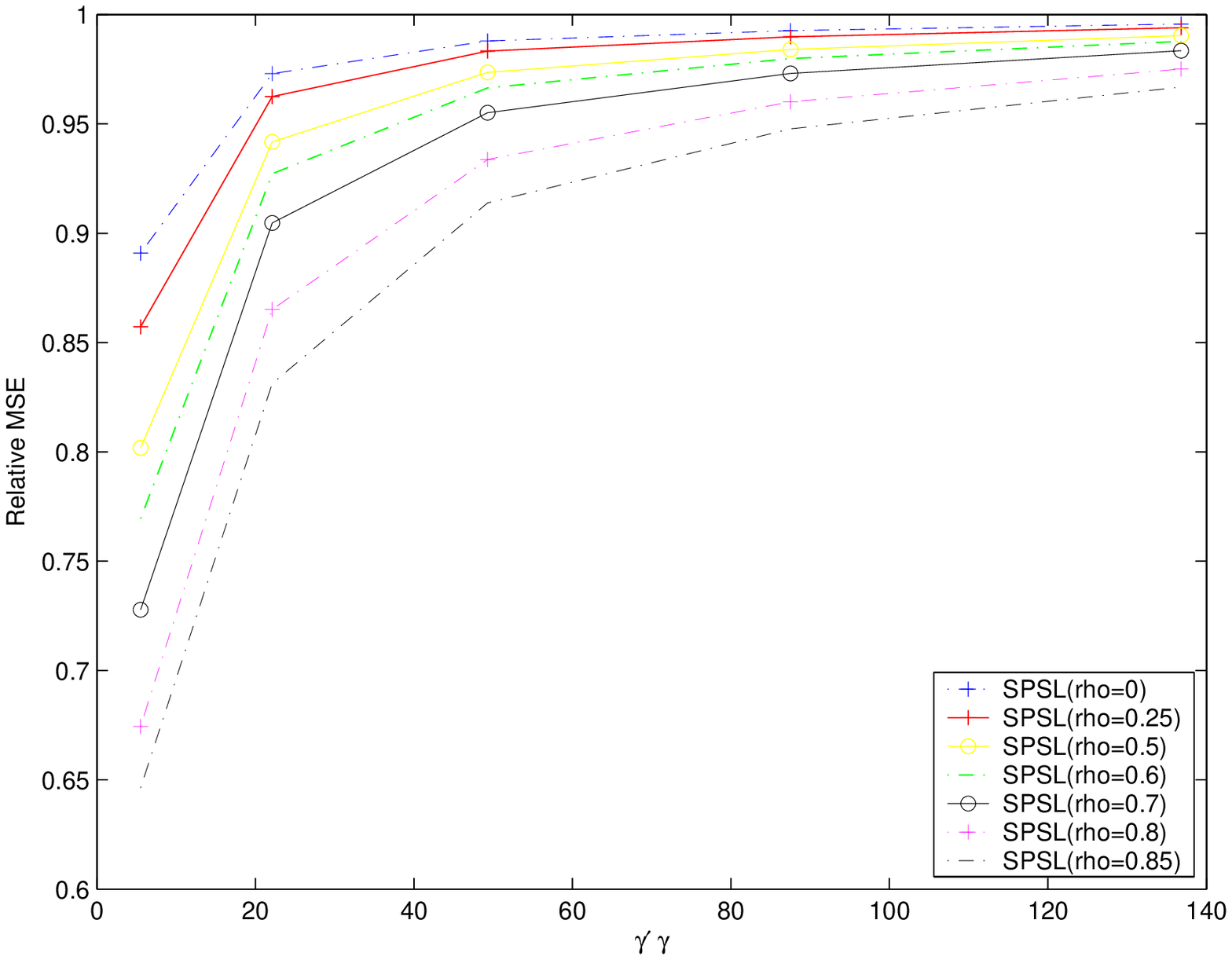}
} \subfigure[\mbox{$n=25$,  $\sigma=0.1$, $k=3$}]{
\label{T15p6} 
\includegraphics[height=1.65in,width=2.75in]{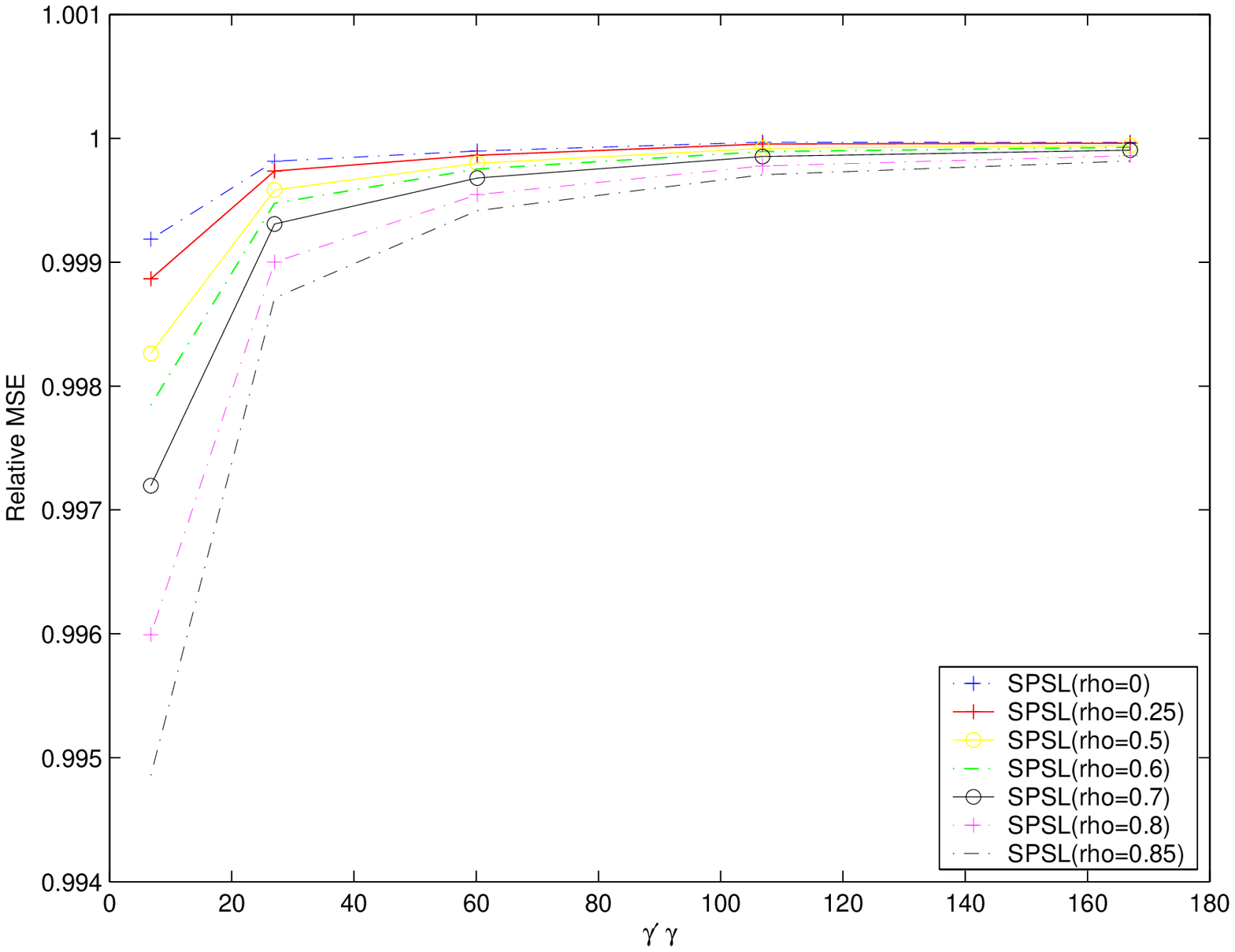}
} \subfigure[\mbox{$n=25$, $\sigma=0.25$, $k=3$} ]{
\label{T30p6} 
\includegraphics[height=1.65in,width=2.75in]{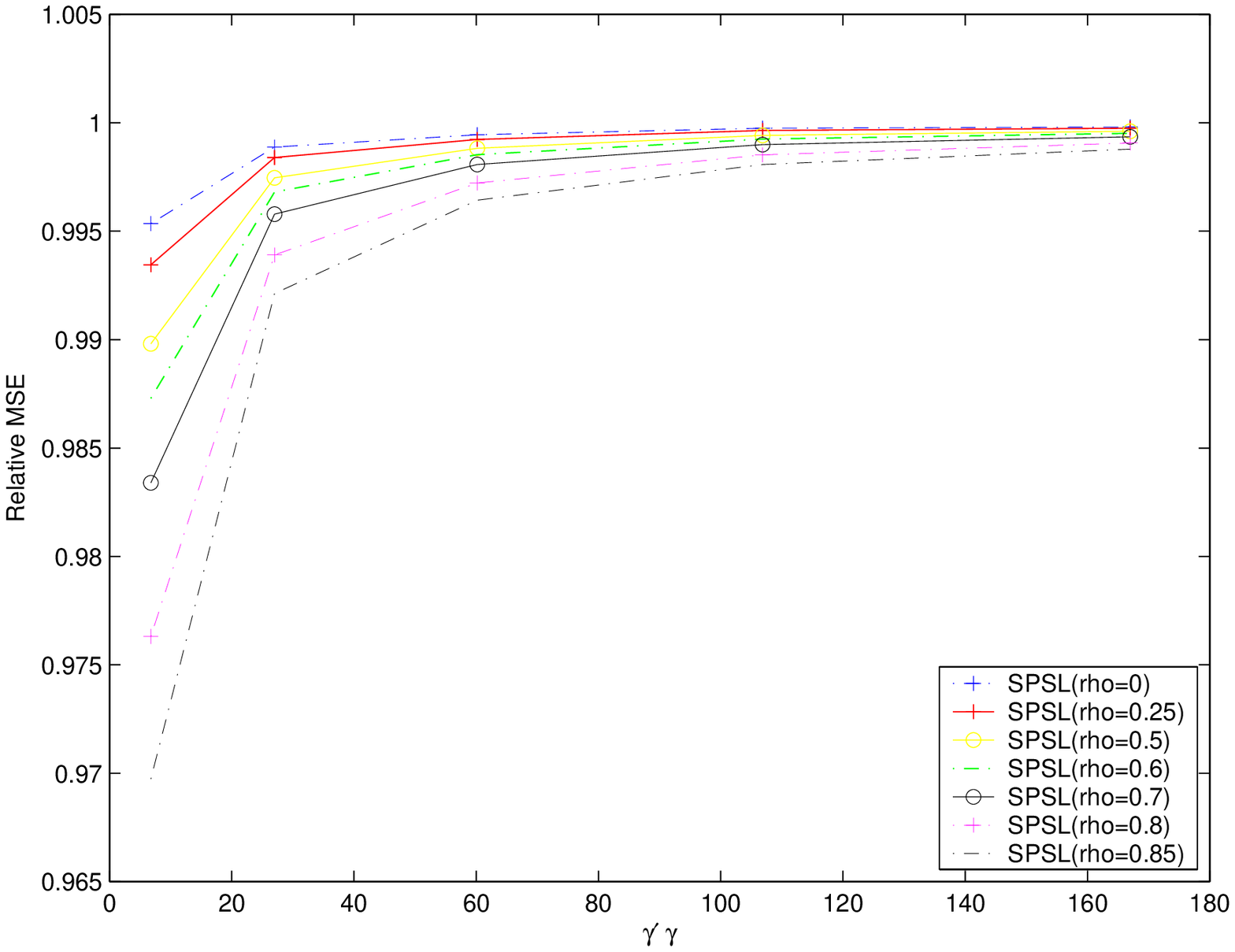}
} \subfigure[\mbox{$n=25$, $\sigma=0.5$, $k=3$}]{
\label{T40p6} 
\includegraphics[height=1.65in,width=2.75in]{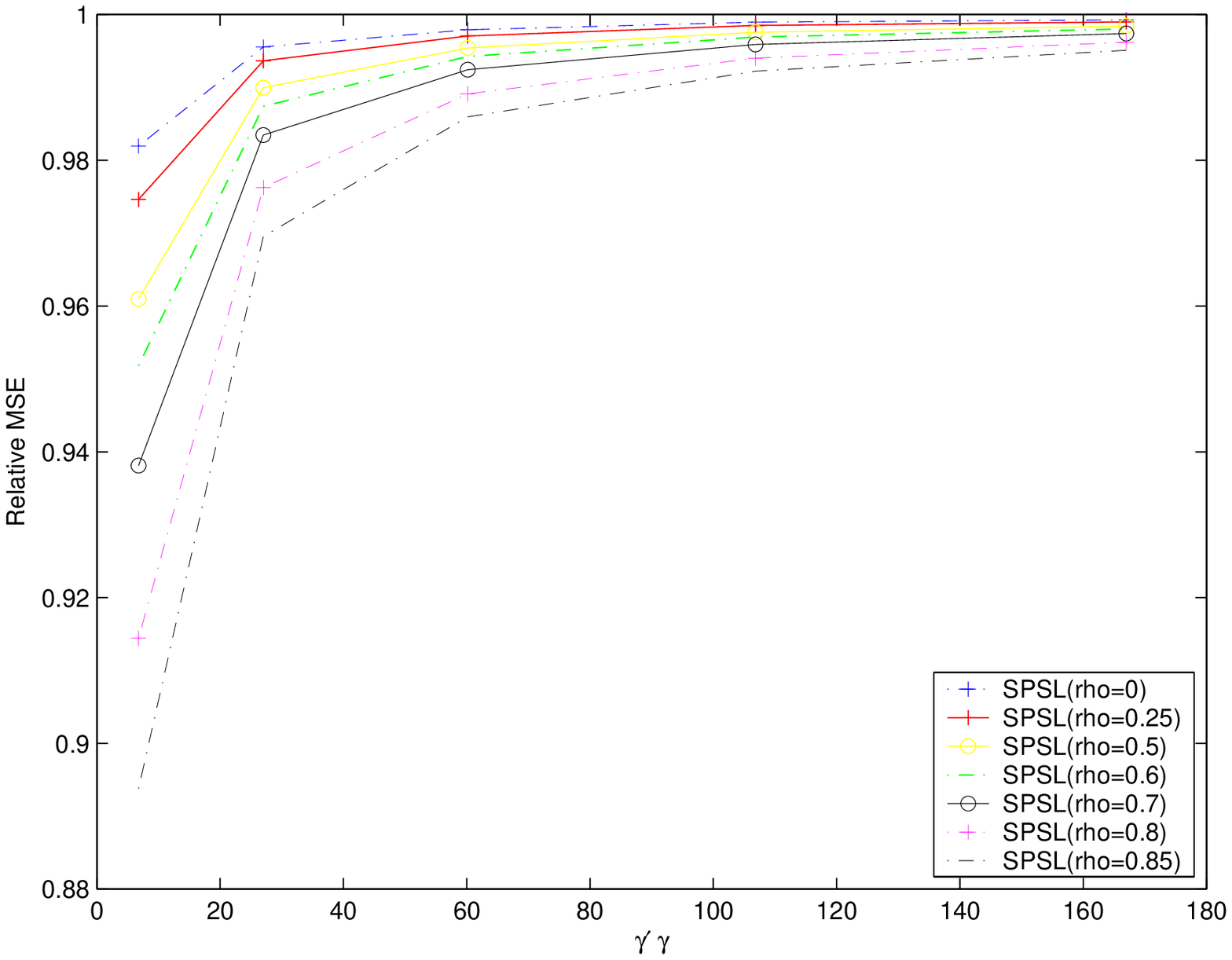}
} \subfigure[\mbox{$n=25$, $\sigma=1$, $k=3$}]{
\label{T50p6} 
\includegraphics[height=1.65in,width=2.75in]{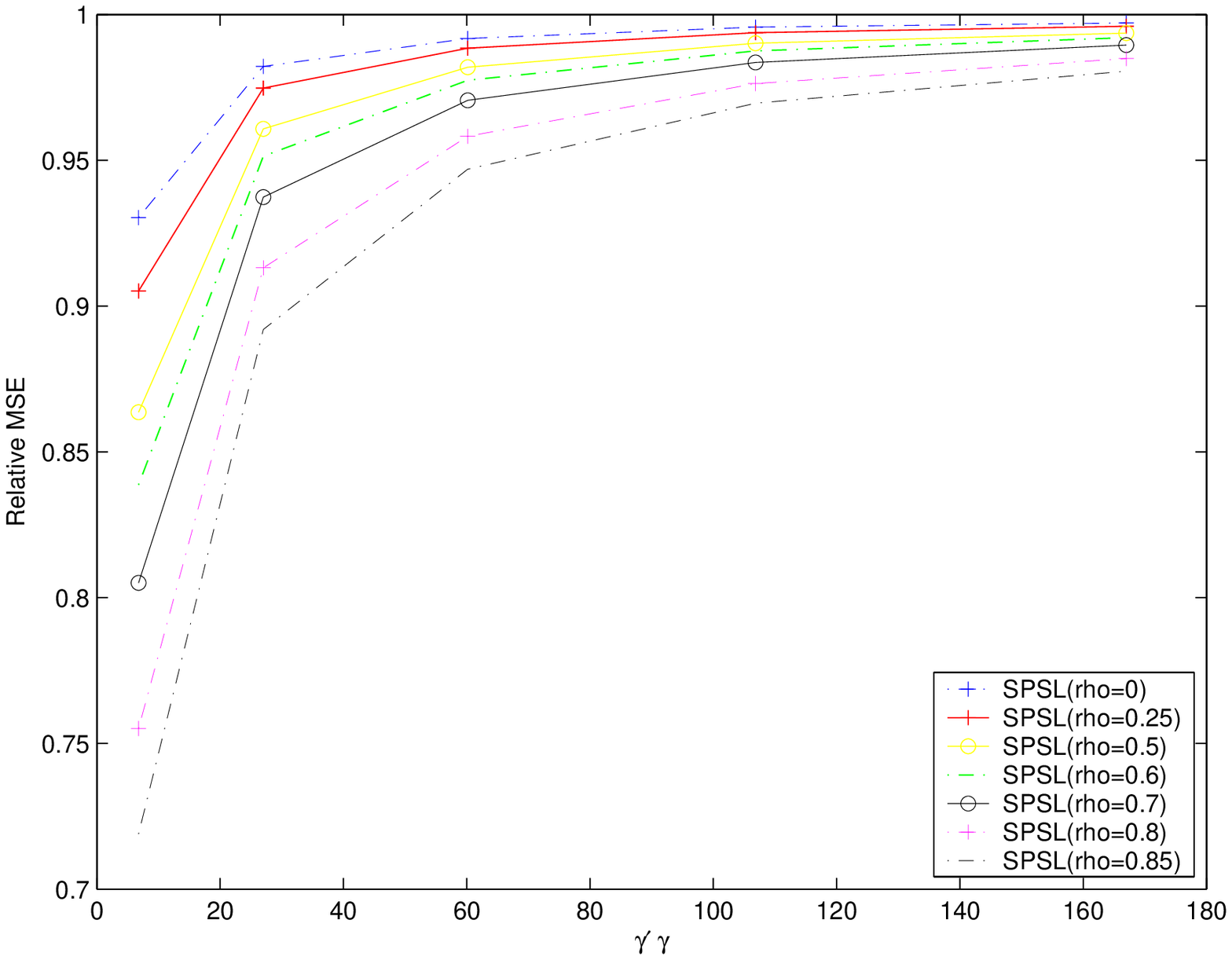}
}
 \caption{Relative efficiency versus \mbox{$\gamma'\gamma$ }}
\label{figloopgamm}
\end{figure}

\begin{figure}[htbp]
\centering
\subfigure[\mbox{$n=15$,  $\sigma=0.1$, $k=4$}]{
\label{T15p6} 
\includegraphics[height=1.65in,width=2.75in]{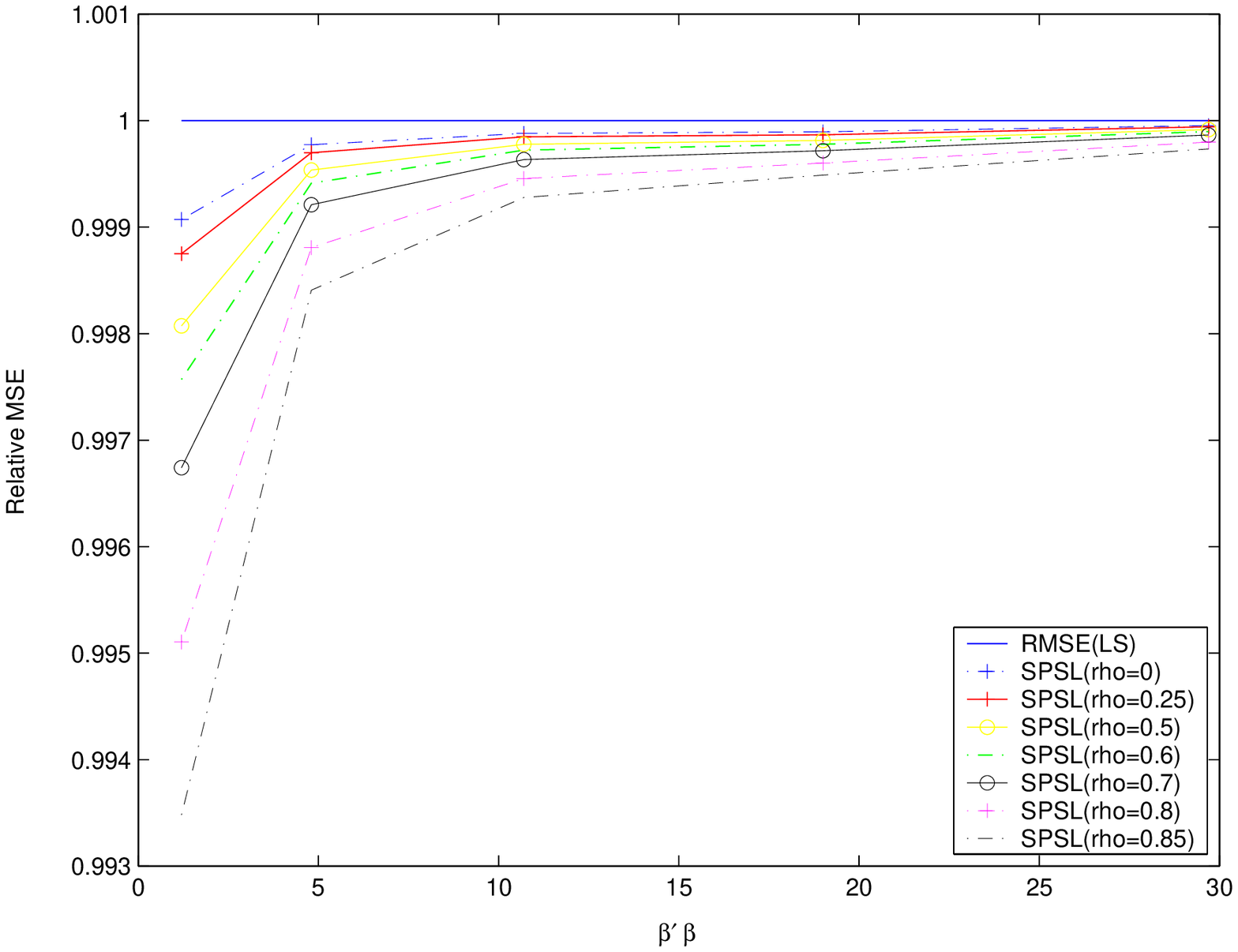}
} \subfigure[\mbox{$n=15$, $\sigma=0.25$, $k=4$} ]{
\label{T30p6} 
\includegraphics[height=1.65in,width=2.75in]{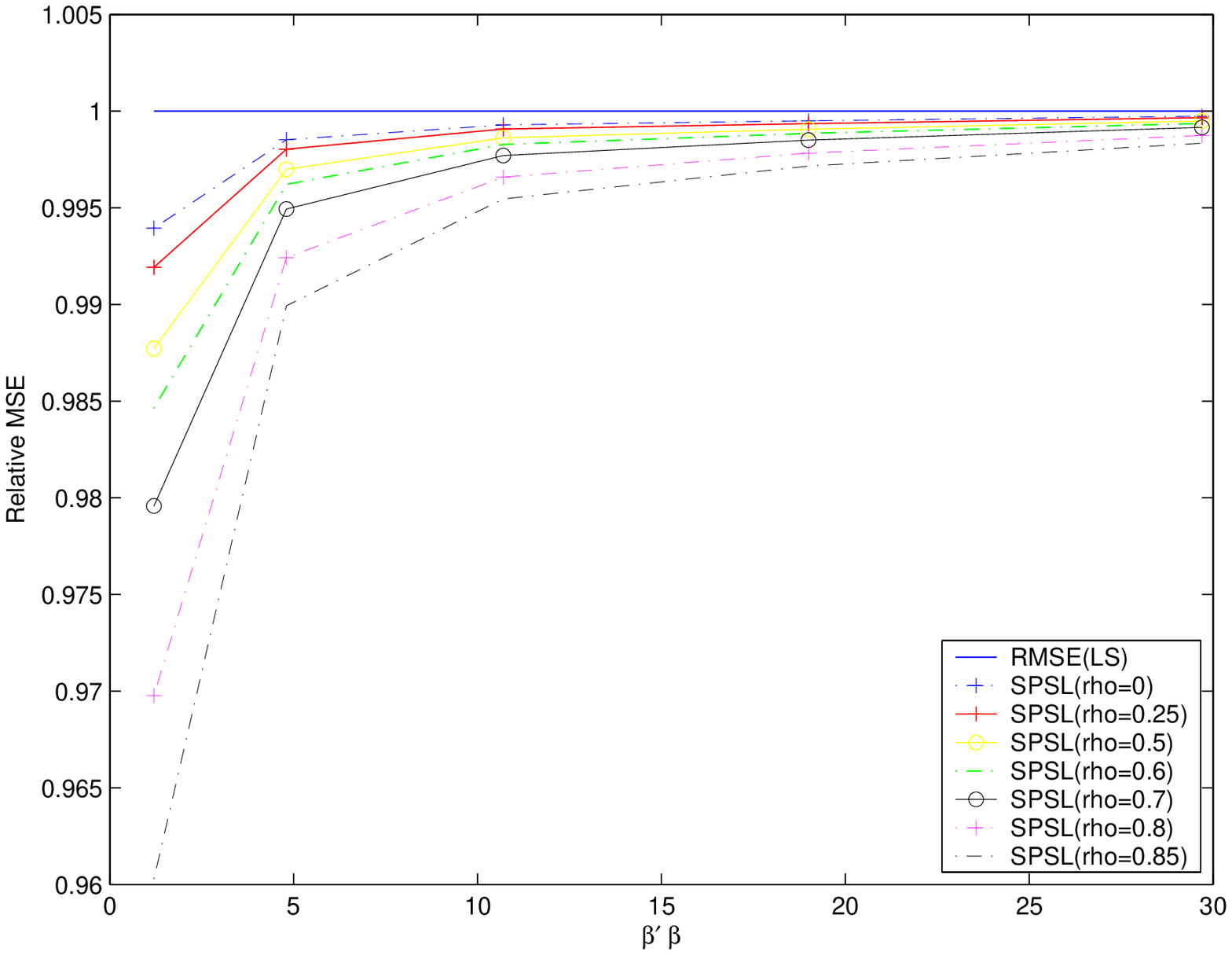}
} \subfigure[\mbox{$n=15$, $\sigma=0.5$, $k=4$}]{
\label{T40p6} 
\includegraphics[height=1.65in,width=2.75in]{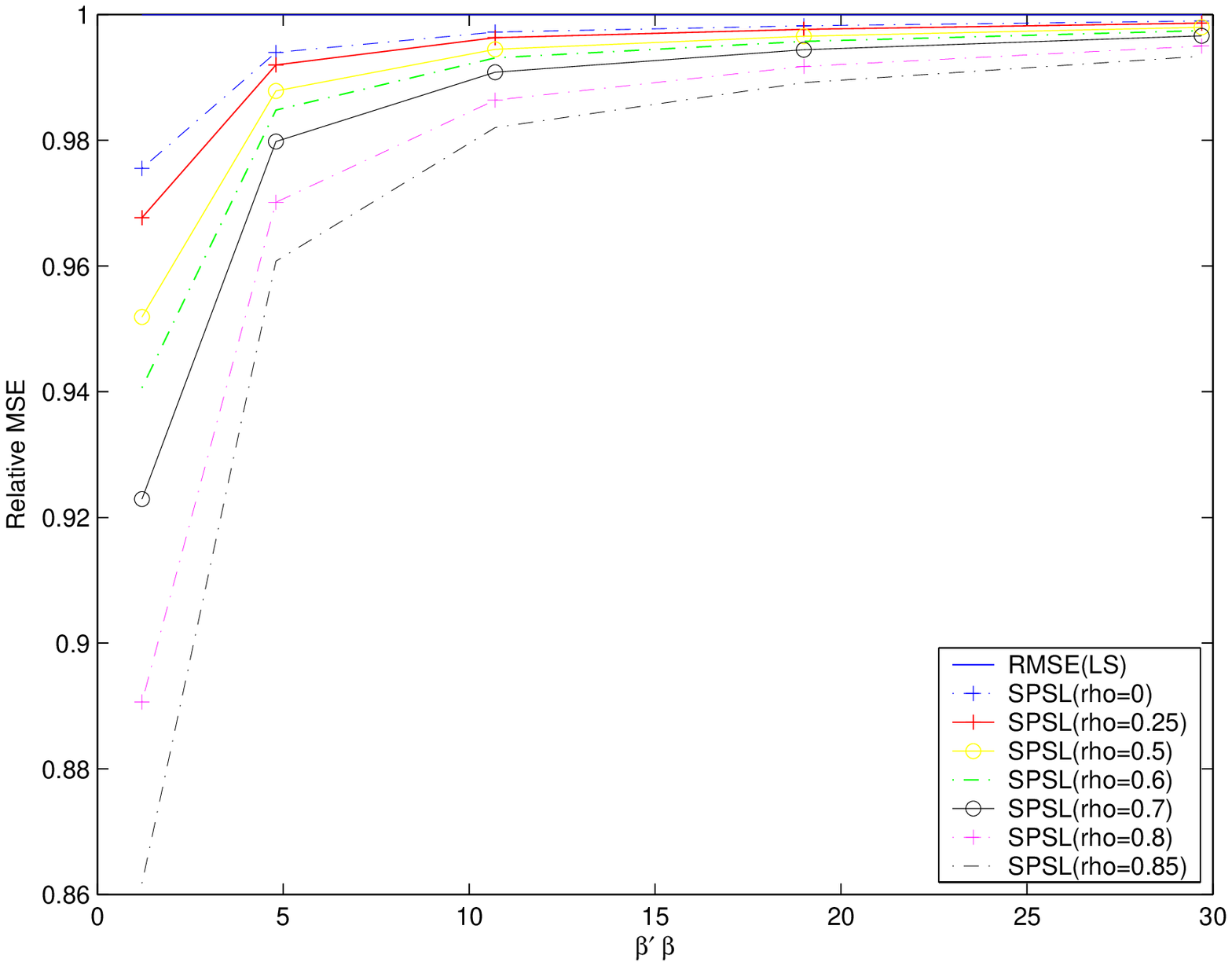}
} \subfigure[\mbox{$n=15$, $\sigma=1$, $k=4$}]{
\label{T50p6} 
\includegraphics[height=1.65in,width=2.75in]{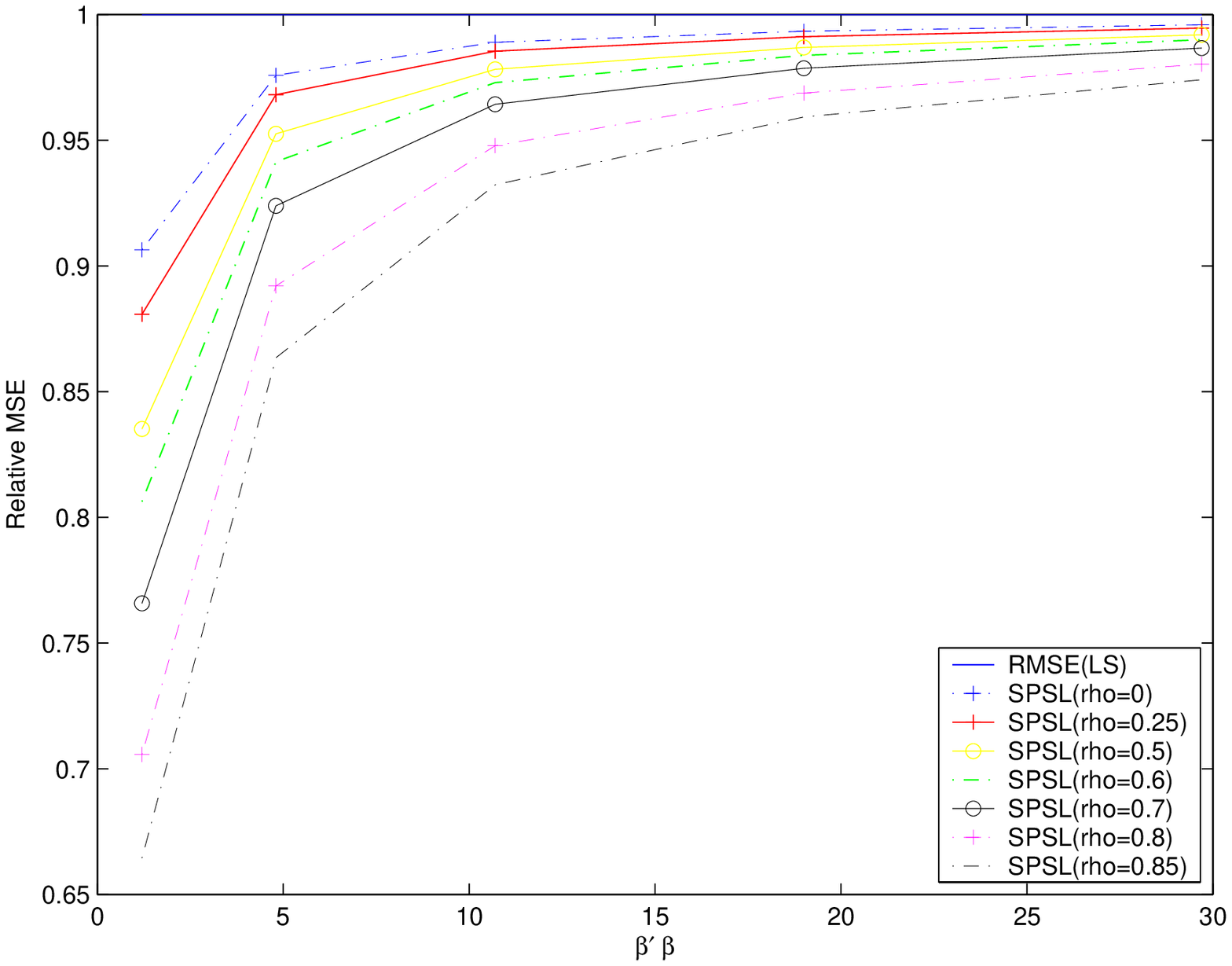}
} \subfigure[\mbox{$n=25$,  $\sigma=0.1$, $k=4$}]{
\label{T15p6} 
\includegraphics[height=1.65in,width=2.75in]{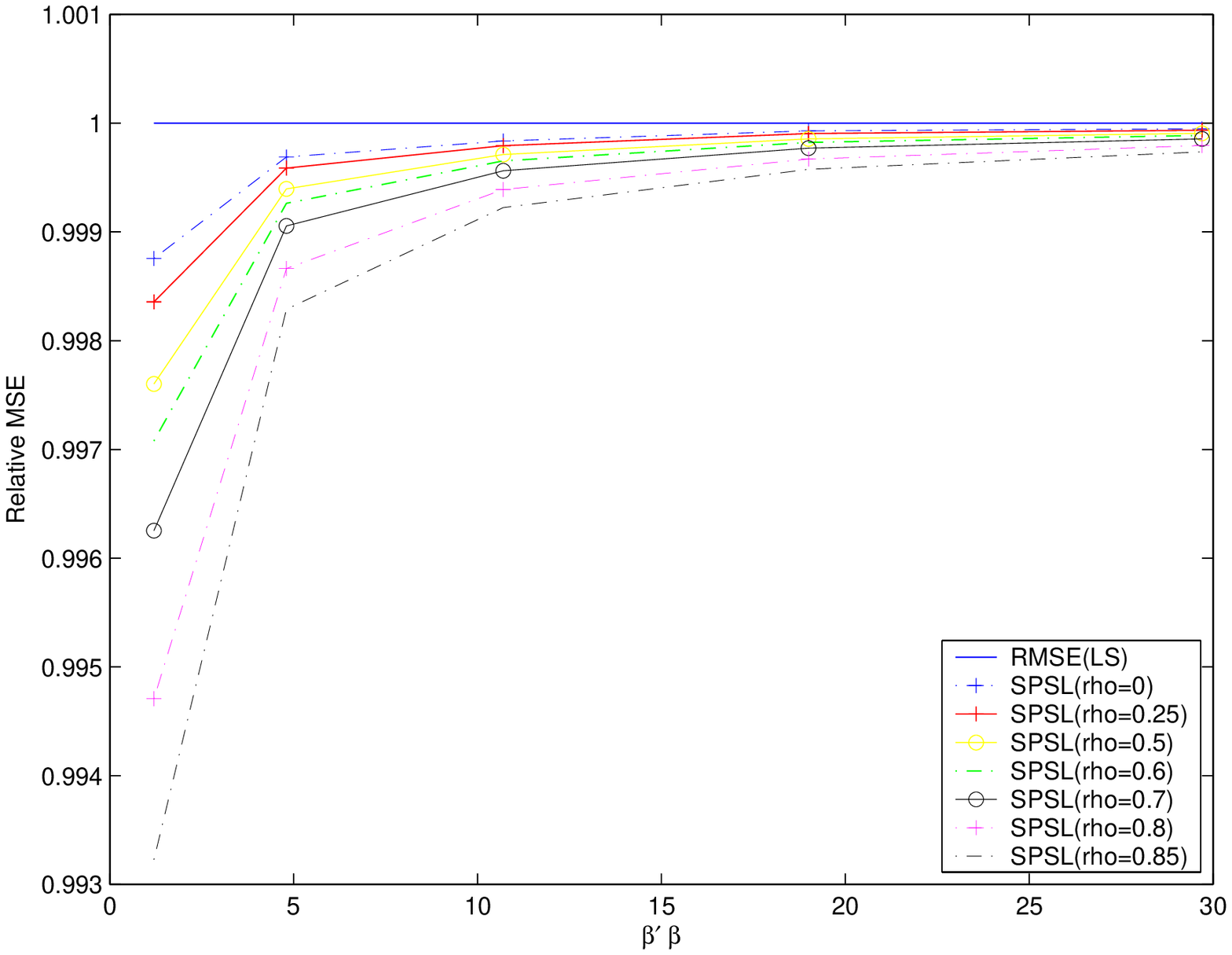}
} \subfigure[\mbox{$n=25$, $\sigma=0.25$, $k=4$} ]{
\label{T30p6} 
\includegraphics[height=1.65in,width=2.75in]{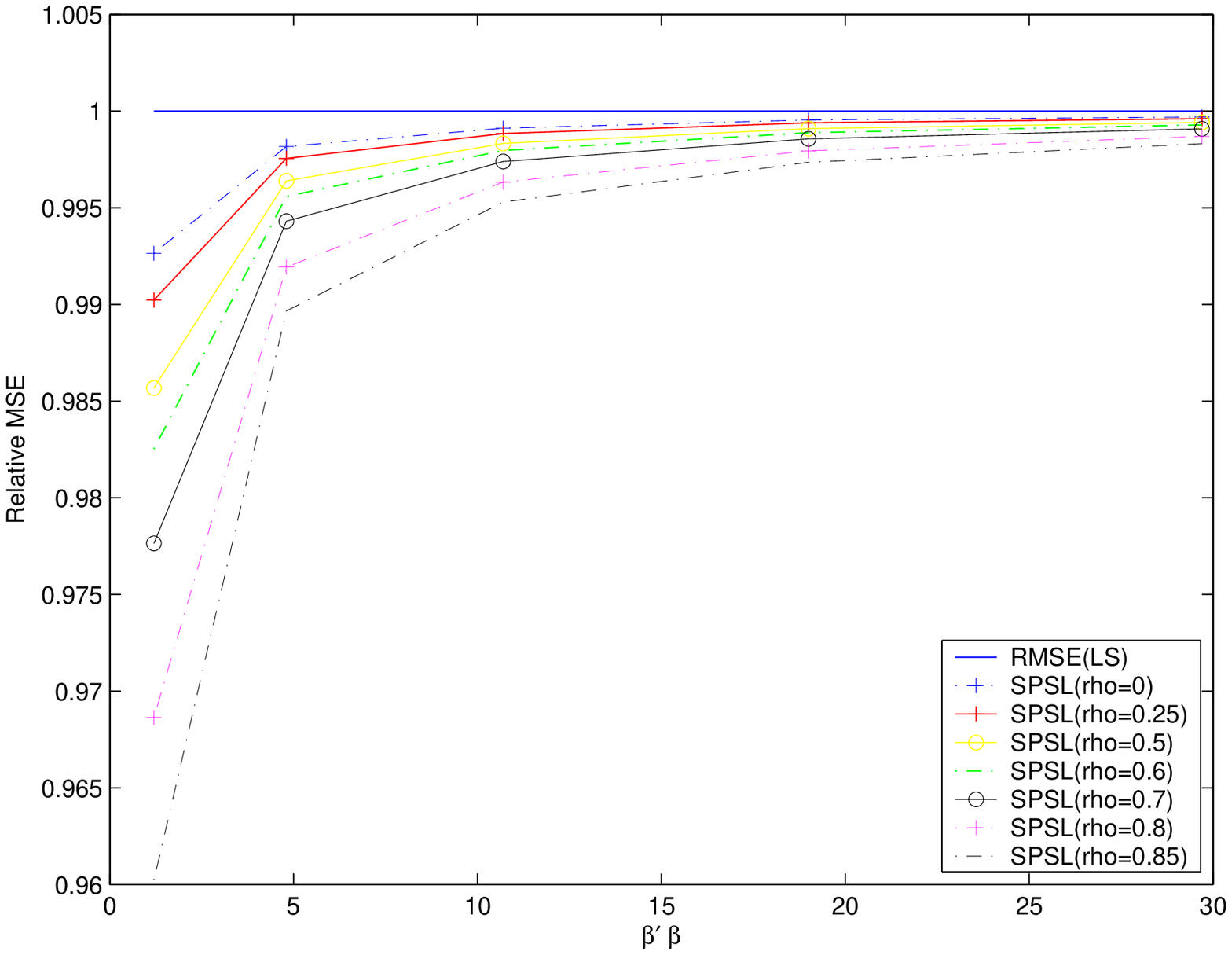}
} \subfigure[\mbox{$n=25$, $\sigma=0.5$, $k=4$}]{
\label{T40p6} 
\includegraphics[height=1.65in,width=2.75in]{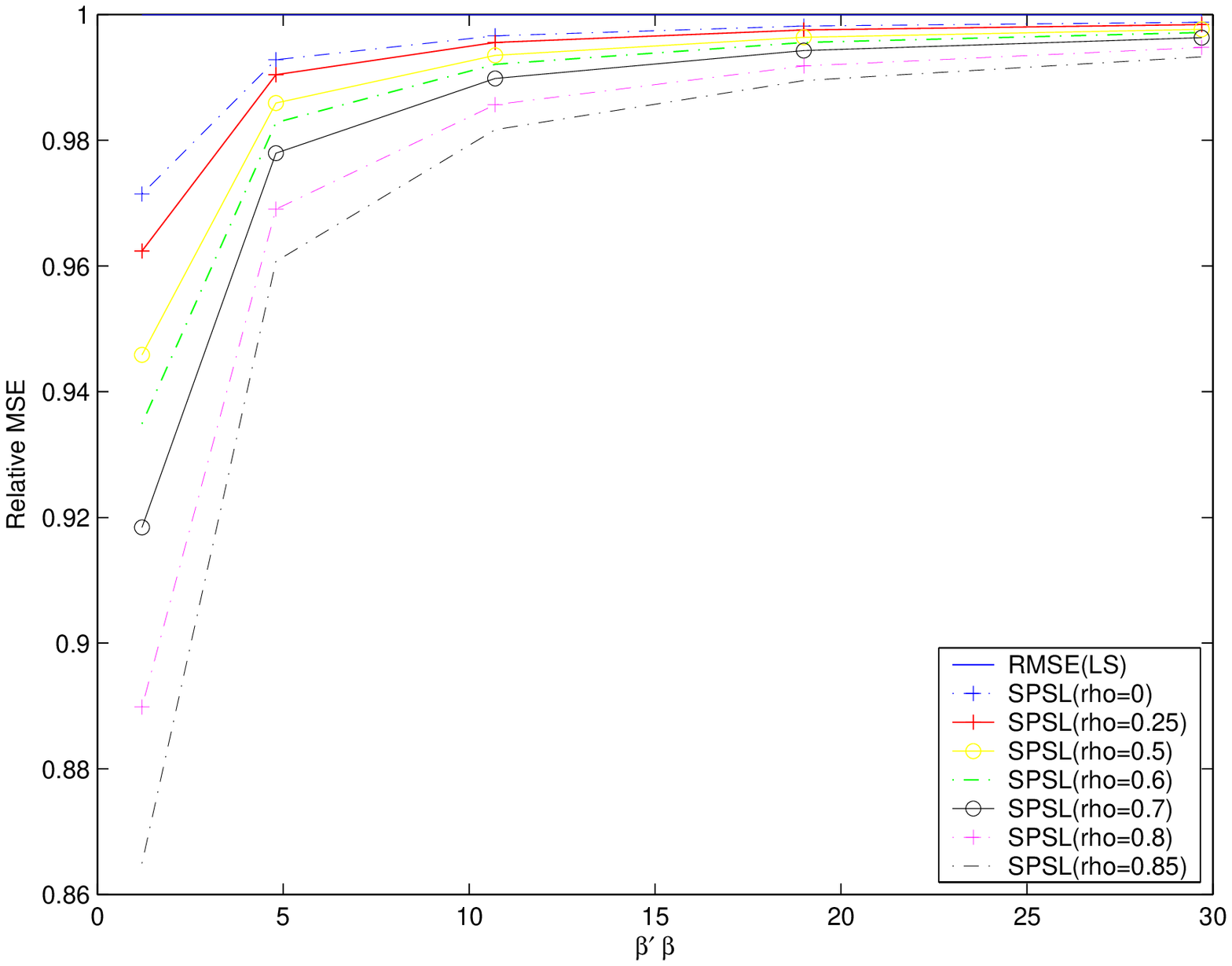}
} \subfigure[\mbox{$n=25$, $\sigma=1$, $k=4$}]{
\label{T50p6} 
\includegraphics[height=1.65in,width=2.75in]{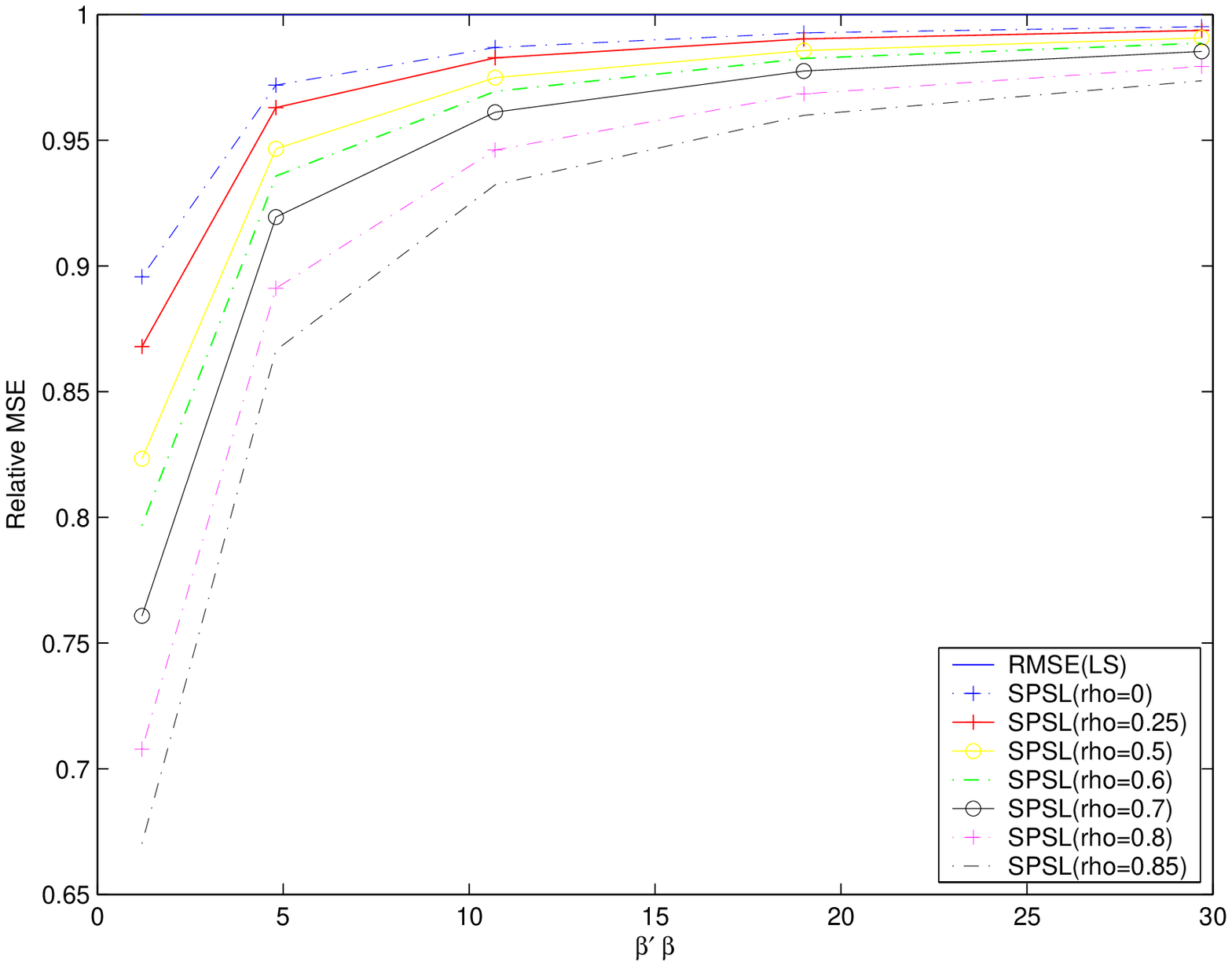}
} \caption{Relative efficiency versus \mbox{$\beta'\beta$ }}
\label{figloopp3}
\end{figure}

\begin{figure}[htbp]
\centering
\subfigure[\mbox{$n=15$,  $\sigma=0.1$, $k=4$}]{
\label{T15p6} 
\includegraphics[height=1.65in,width=2.75in]{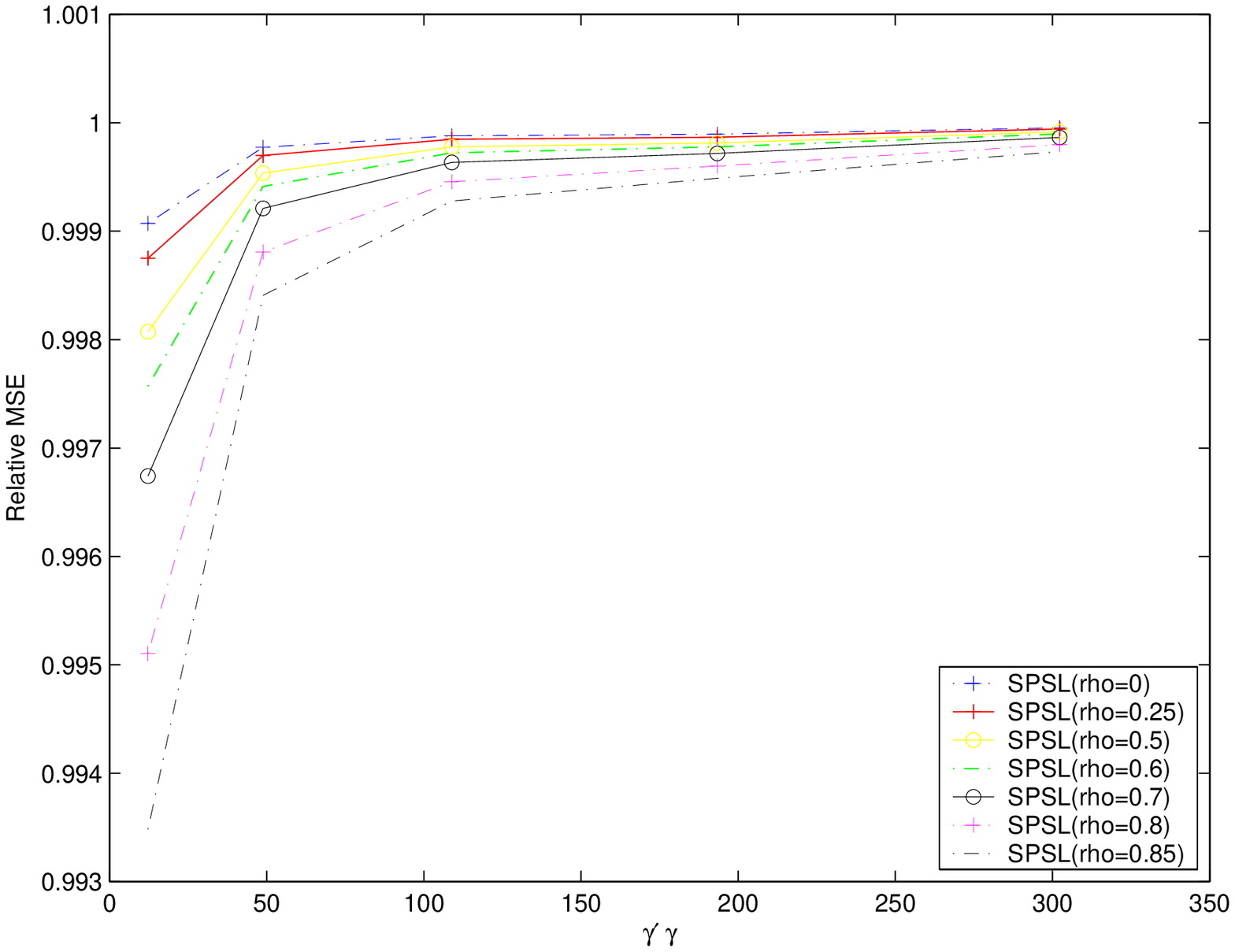}
} \subfigure[\mbox{$n=15$, $\sigma=0.25$, $k=4$} ]{
\label{T30p6} 
\includegraphics[height=1.65in,width=2.75in]{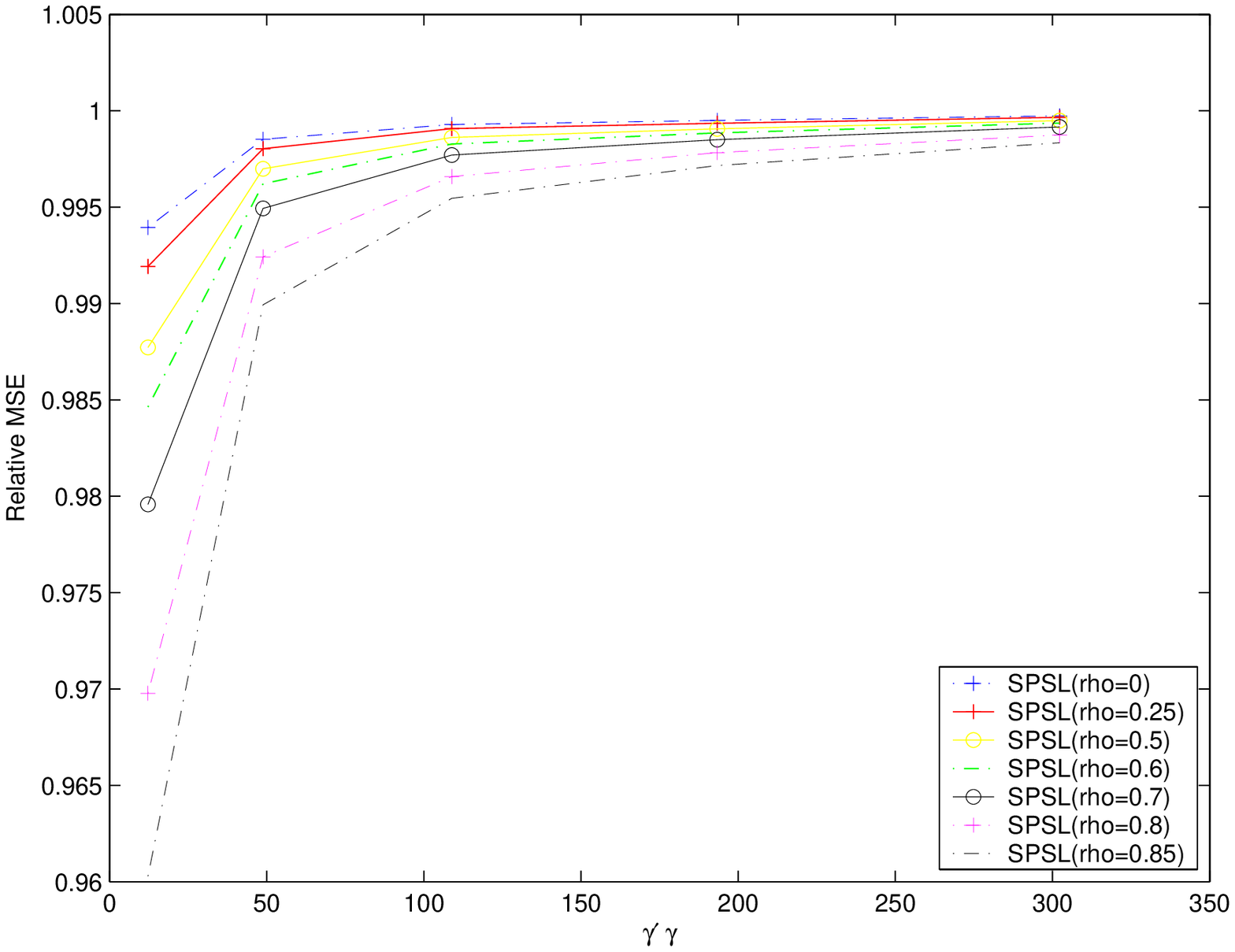}
} \subfigure[\mbox{$n=15$, $\sigma=0.5$, $k=4$}]{
\label{T40p6} 
\includegraphics[height=1.65in,width=2.75in]{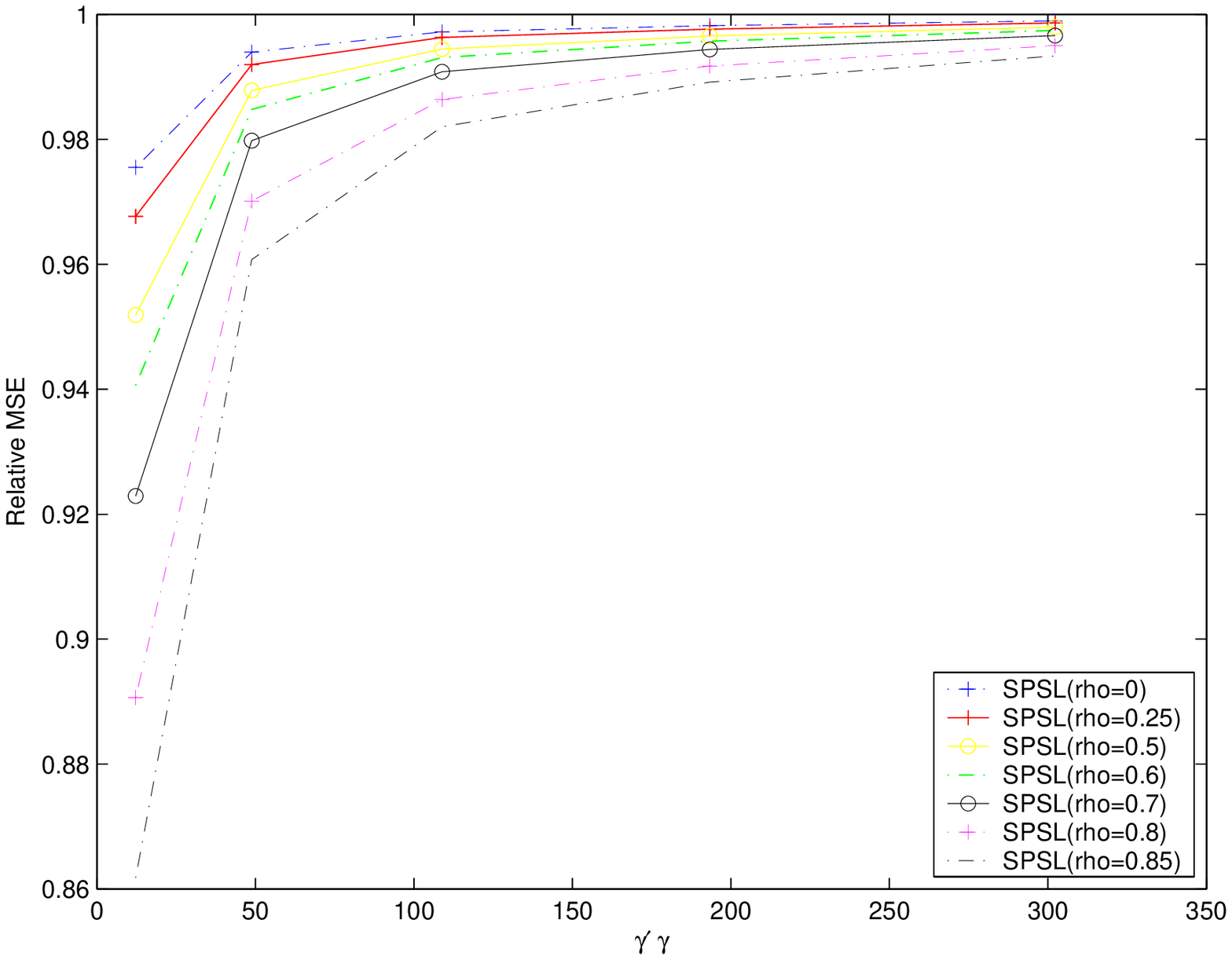}
} \subfigure[\mbox{$n=15$, $\sigma=1$, $k=4$}]{
\label{T50p6} 
\includegraphics[height=1.65in,width=2.75in]{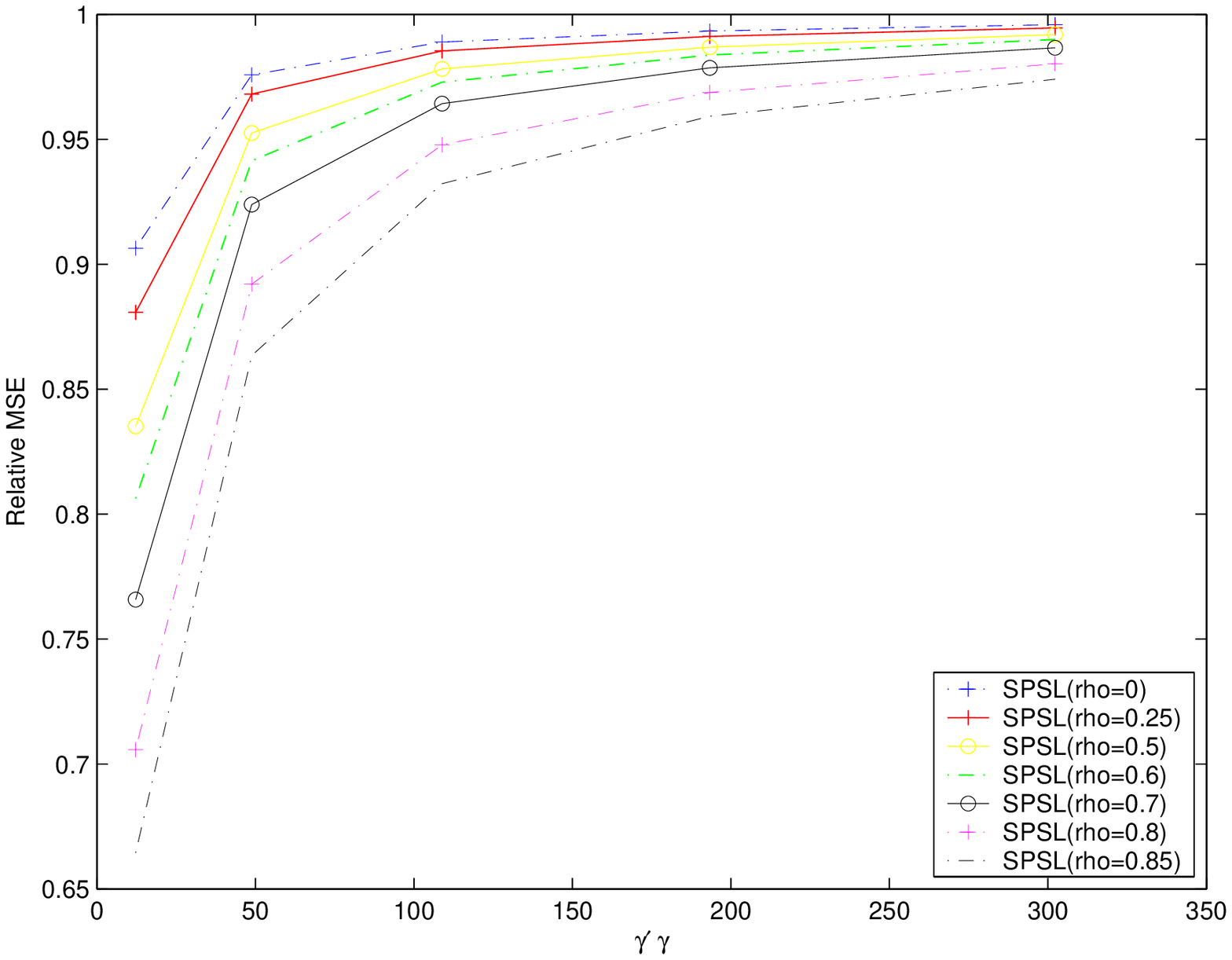}
} \subfigure[\mbox{$n=25$,  $\sigma=0.1$, $k=4$}]{
\label{T15p6} 
\includegraphics[height=1.65in,width=2.75in]{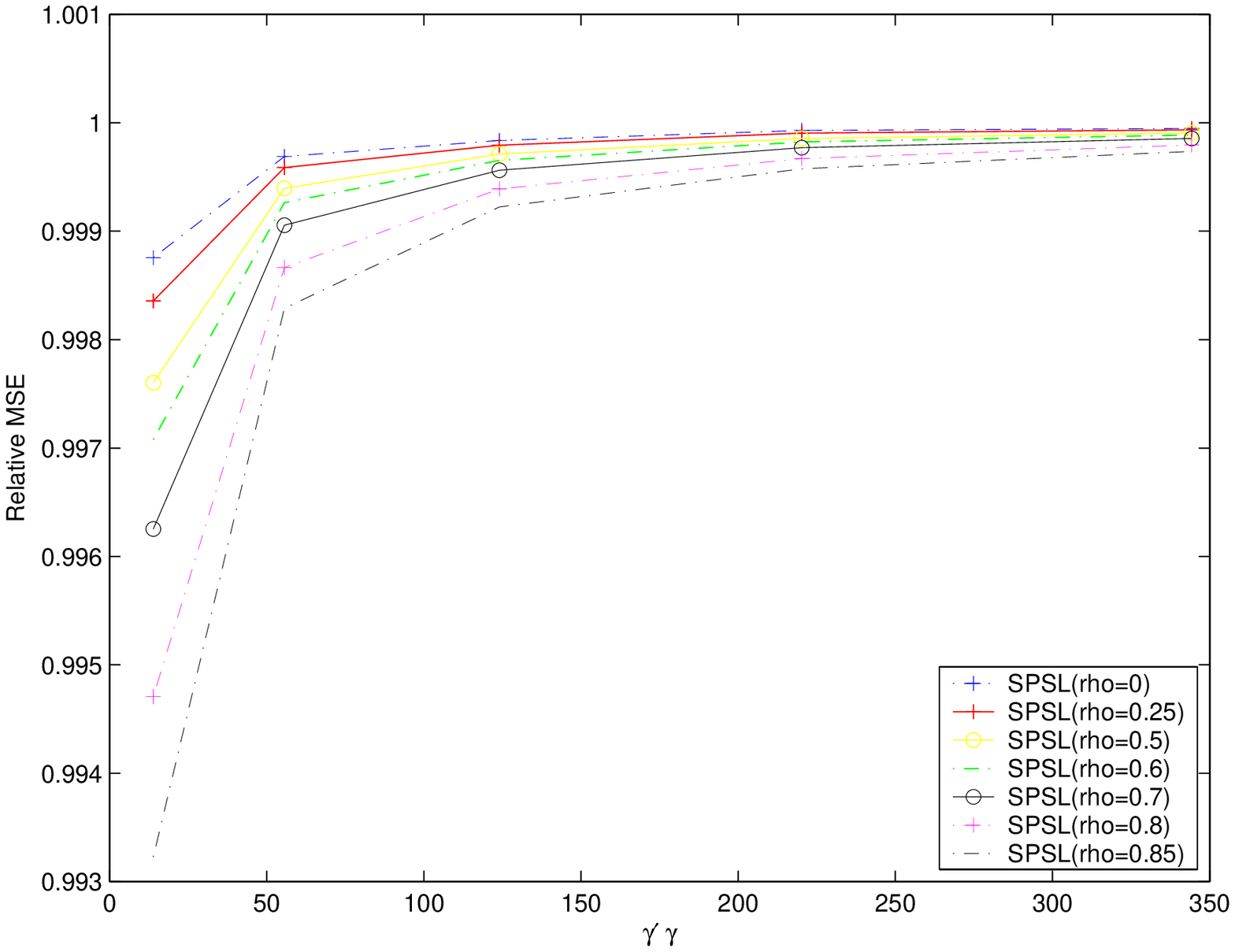}
} \subfigure[\mbox{$n=25$, $\sigma=0.25$, $k=4$} ]{
\label{T30p6} 
\includegraphics[height=1.65in,width=2.75in]{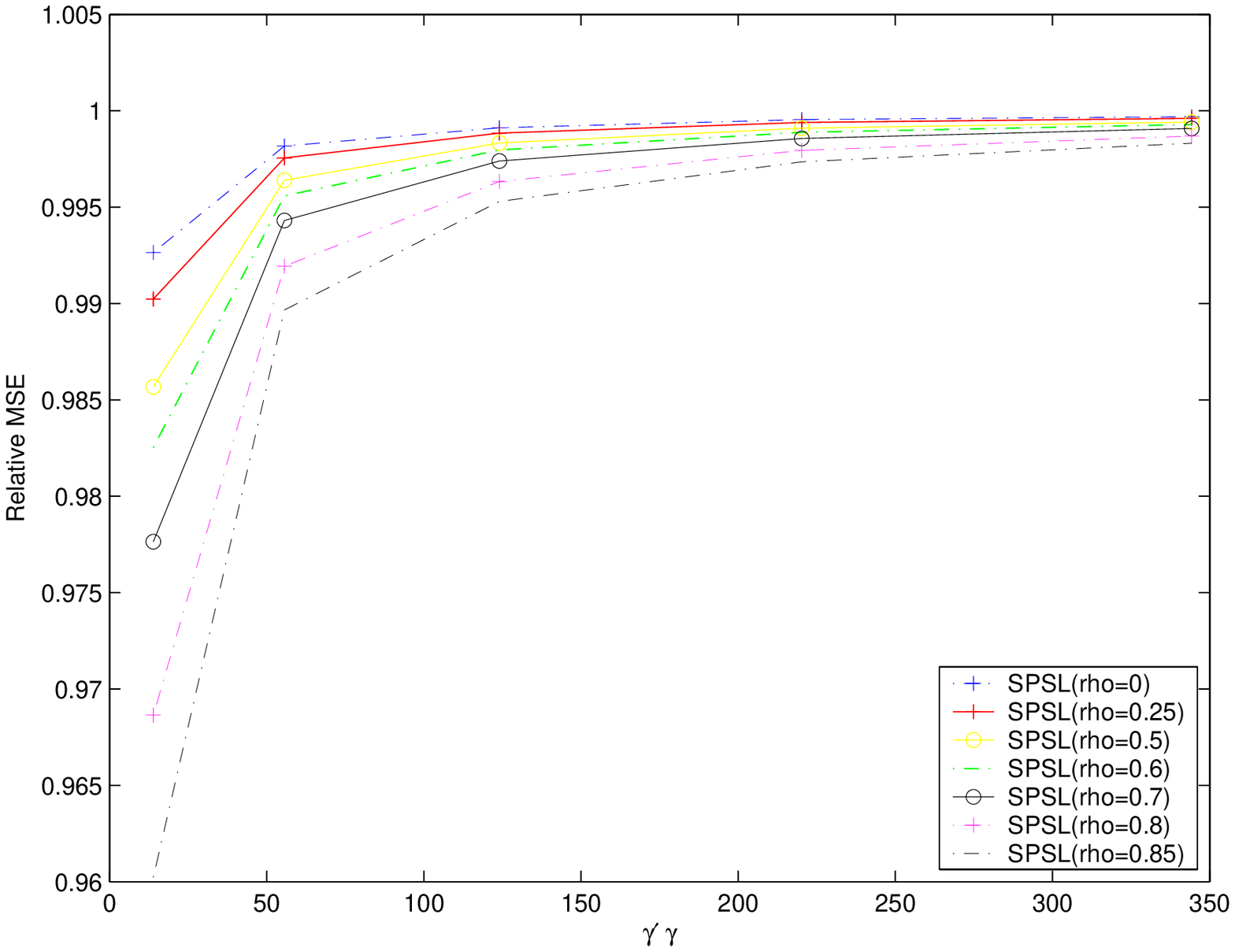}
} \subfigure[\mbox{$n=25$, $\sigma=0.5$, $k=4$}]{
\label{T40p6} 
\includegraphics[height=1.65in,width=2.75in]{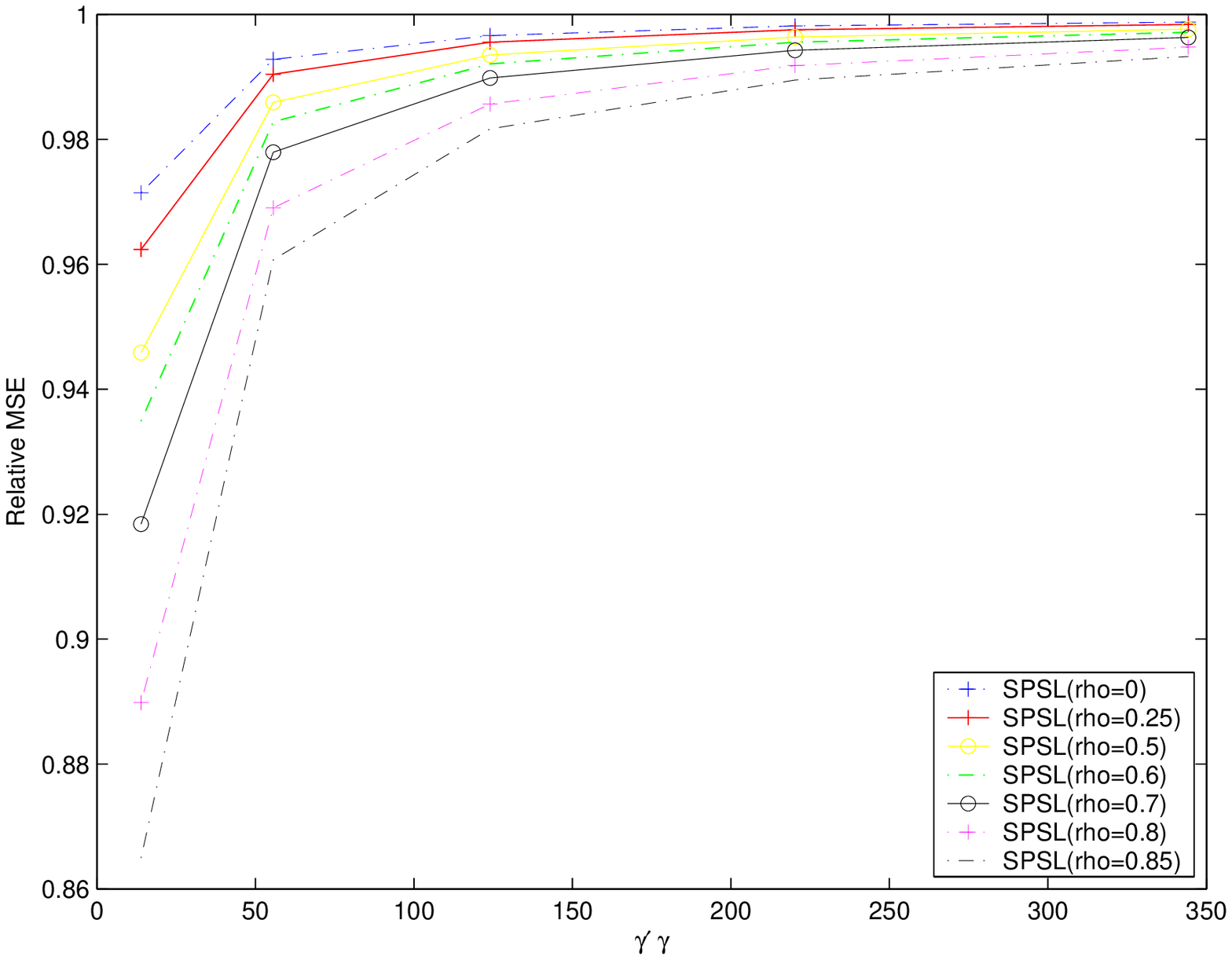}
} \subfigure[\mbox{$n=25$, $\sigma=1$, $k=4$}]{
\label{T50p6} 
\includegraphics[height=1.65in,width=2.75in]{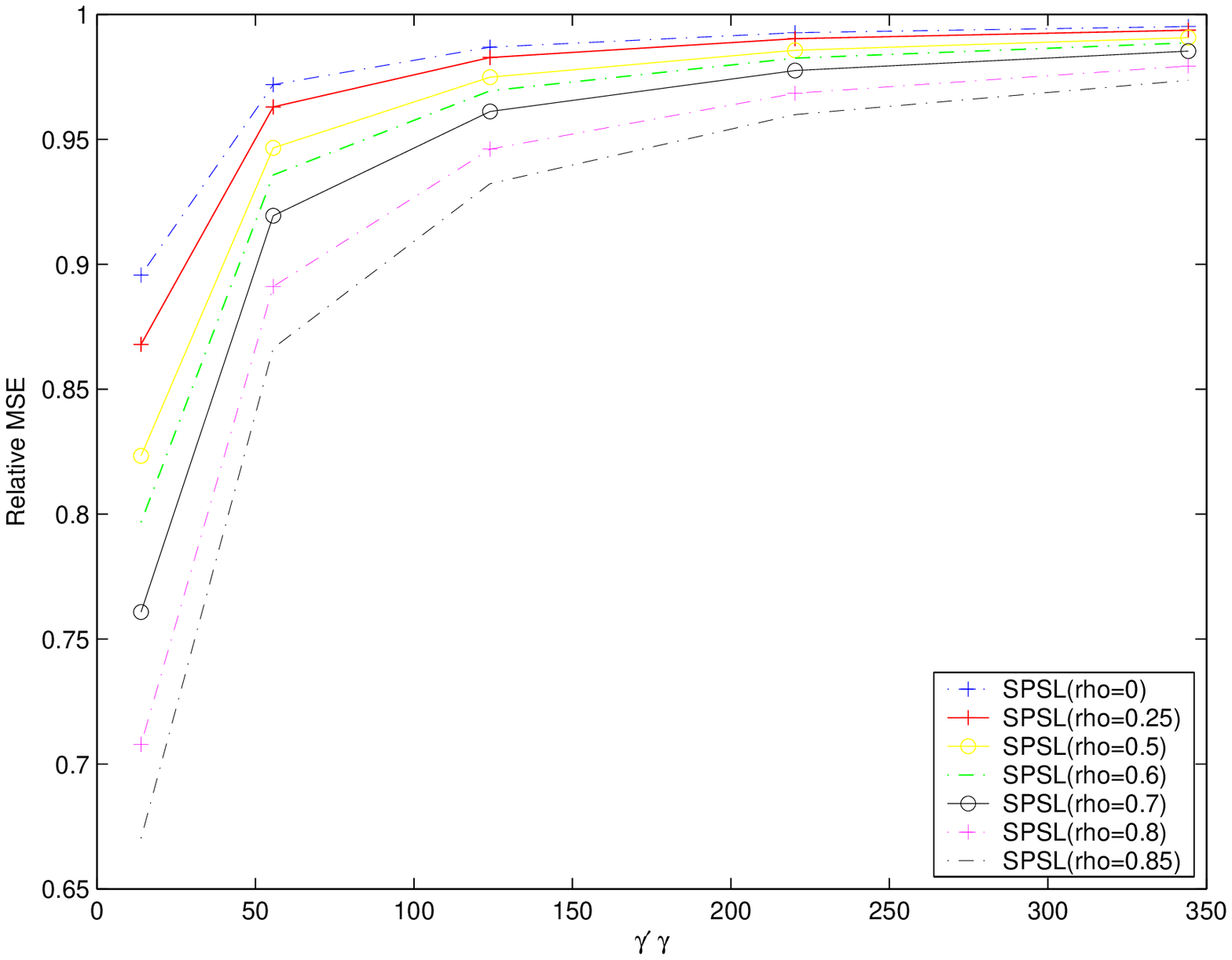}
} \caption{Relative efficiency versus \mbox{$\gamma'\gamma$ }}
\label{figloopp3bis}
\end{figure}

More precisely, in this paper, by following the results in
Subsection~5.2 of JM~(2004), we examine the relative performance of
the SPLS estimator as a function of the parameter norms
$\beta'\beta$ and $\gamma'\gamma$, where the parameter vector
$\beta$ is chosen such that \\$\beta'\beta\in \left\{1.2, 4.8, 10.7,
19.0, 29.7\right\}$, which are the values used in JM~(2004,
Figure~3).

For the small sample sizes and $k=3$,  the results are presented in
Figures~\ref{figloopbeta} and~\ref{figloopgamm}. These figures show
that, as the norms $\beta'\beta$ or $\gamma'\gamma$ increase, the
risk of the SPSL estimator increases and approaches the risk of the
LS estimator. Further, Figures~\ref{figloopp3}-\ref{figloopp3bis} 
show a
similar pattern for the case where $k=4$ and/or for the cases of
moderate and large sample sizes. This result is in agreement with
that given in JM~(2004) for the cases where $k\geqslant5$. Also,
these figures confirm the findings in JM~(2004) in that the
correlation among the $X$ variables increases the relative
performance of the SPSL estimator over the LS estimator.

\subsection{Data analysis}\label{sec:realdata} In this subsection, we
illustrate the application of the proposed method to a real data
set.  The data set consists of  a sample of 25 brands of cigarettes
(see Mendenhall  and Sincich,~1992). For each brand of
cigarette, the measurements of weight as well as tar, nicotine, and carbon
monoxide content are have been recorded.

The choice of this data set is justified by several health and
environmental  issues concerning the cigarettes as mentioned by some medical
studies. Thus, as explained in Mendenhall  and Sincich~(1992),  ''\emph{the
United States Surgeon General considers each of these substances
hazardous to a smoker's health}''. The authors  also
mention that ''\emph{past studies have shown that increases in the
tar and nicotine content of a cigarette are accompanied by an
increase in the carbon monoxide emitted from the cigarette smoke}''.

Accordingly, in order to illustrate the application of the proposed method,
the response variable is taken as the carbon monoxide content, while
the three covariates are:  $X_{1}$: \emph{weight}; $X_{2}$:
\emph{tar content}, and $X_{3}$: \emph{nicotine content}. So,
including the intercept, we apply the proposed method to the
regression model for which $n=25$ and $k=4$. It should be noticed
that, for such a data set whose $k<5$, the result in JM~(2004)
cannot be used to justify the efficiency of the SPSL over the base
estimator. In contrast, the result established in this paper justifies
very well the relative efficiency of the SPSL estimator provided that the
underlying distribution of the error terms is an
elliptically contoured distribution. To give some numerical
descriptive measures, the sample mean is 12.5280 for the response
variable, and for the covariates, the sample means are 12.2160, \quad{ }
0.8764 \mbox{ } and \mbox{ } 0.9703 for the weight, tar and nicotine
content, respectively. The correlation coefficients between the
response and the covariates are shown in Table~\ref{correlation}.
This table indicates that the weight and the tar content are highly
correlated to the response, while the correlation between the
nicotine content and the response is modest. Nevertheless,  the
correlation coefficients are statistically significant at the 5~\% level.
Further, the covariates seem pairwise correlated at
significance level 5~\%.
\begin{table}[!htbp]
\begin{center}
\caption{Correlation between covariates and  response variables
(p-value in parentheses)}
\begin{tabular}{l*{5}{c}r}
\hline
     & carbon monoxide  & Weight & Tar content & Nicotine content\\
\hline
carbon monoxide    & 1 & 0.9575  & 0.9259 &  0.4640  \\
               & (\small{-}) &   (\small{0.0000})    &   (\small{0.0000 })    &   (\small{0.0195})  \\
\hline
Weight    & 0.9575 & 1  & 0.9766 &  0.4908  \\
          &(\small{0.0000}) &   (\small{-})    &   (\small{0.0000 })    &   (\small{0.0127})  \\
\hline
Tar content & 0.9259   & 0.9766  & 1 &  0.5002  \\
            & (\small{0.0000 })  &   (\small{0.0000})    &   (\small{-})    &   (\small{0.0109})  \\
\hline
Nicotine content &0.4640    & 0.4908  & 0.5002  &  1  \\
               & (\small{0.0195}) &   (\small{0.0127 })    &   (\small{ 0.0109 })    &   (\small{-})  \\

\hline

\end{tabular}
\label{correlation}
\end{center}
\end{table}
By applying the method, we obtain the point estimates based on the LS and
SPSL estimators as reported in Table~\ref{dataestimate}. To asses the performance
of the estimators, we compute the mean squared error based on a
bootstrap method with 5000 replications. The relative efficiency of
the estimators is given in Table~\ref{efficiency}.
\begin{table}[!htbp]
\begin{center}
\caption{Point Estimates}
\begin{tabular}{ccc}
\hline
Parameter           & LS  & SPSL\\
\hline
Intercept   &  3.2022  &  3.9325\\
Weight      & 0.9626    &  0.9645\\
 Tar      &  -2.6317  &  -1.3262 \\
Nicotine & -0.1305 &  0.8983 \\
\hline
\end{tabular}
\label{dataestimate}
\end{center}
\end{table}

\begin{table}[!htbp]
\begin{center}
\caption{Relative efficiency (Bootstrap)}
\begin{tabular}{ccc}
\hline
Estimator           & LS  & SPSL\\
\hline
Relative efficiency   &  1  &  0.7578\\
\hline
\end{tabular}
\label{efficiency}
\end{center}
\end{table}

From Table~\ref{efficiency}, one can clearly see that the relative
efficiency of the SPSL estimstor is less than 1, that is the relative
efficiency of the base estimator. This illustrates that SPSL
dominates the base estimator.
\section{Conclusion}\label{sec:conclu}
In order to conclude, let us first recall that the main result in JM~(2004)
gives a sufficient condition for the Stein rule~(SR)-type
estimator to dominate the base estimator.
In this paper, we provided more refined inequalities and bounds
which are used in establishing this result in its full generality.
Namely, we generalized in four ways the result in JM~(2004)
which gives a sufficient condition for the Stein rule~(SR)-type
estimator to dominate the base estimator. To this end, we provided
an alternative and a versatile approach for establishing the
dominance result in its full generality. In particular, we proved
theoretically that the nonzero correlation does not change the
condition for the risk dominance of the SR estimator. The impact of
this result is that, unlike the method in JM~(2004), our method is
also applicable to a wide range of regression models, including, for
instance, the quadratic or cubic regressions, where the number of
regressors is less than 5. In addition, we relax the condition of
normality of the sampling distribution of the base estimator.  We
also generalize the method in JM~(2004) to the case where the
variance-covariance matrix of the base estimator and the restricted
estimator may be singular. The significance of this finding is that
our method is also efficient in the case where the past
statistical investigations may have established that some regression
coefficients are not statistically significant while a field
specialist believes that the nonsignificant explanatory variables
are important. From the practical point of view, we evaluate
numerically the relative efficiency of a data-based semiparametric
Stein-like~(SPSL) estimator. The simulation studies corroborate this
theoretical finding that the sufficient condition, for the SPSL to
dominate the LS estimator ($k\geqslant 3$), holds also regardless of
the correlation factor. Nevertheless, Figures~1-8 show that the
correlation may amplify the risk dominance of the SPSL. Finally, the
proposed method is applied to the Cigarette dataset, produced by USA
Federal Trade Commission,  for which $k=4$. An interesting
result is that, by using a bootstrap method, we see that the SPSL
dominates the base estimator. This finding is in agreement with the
theoretical result proved in the present paper.

\appendix{ }
\section{Appendix}
\begin{prn}\label{prncourant}
Let $\bm{C}$ be a $m\times m$-symmetric matrix. Then,\\
 $
 \left|x'Cx\right|\leqslant
\max\{|\lambda_{1}(\bm{C})|,|\lambda_{2}(\bm{C})|,\dots,|\lambda_{m}(\bm{C})|\}x'x$,
for all $m$-column vector $x$.
\end{prn}
\begin{proof}
Since the matrix $\bm{C}$ is symmetric, there exist orthogonal
matrix $\bm{Q}$ such that

$\bm{Q}'\bm{C}\bm{Q}=\bm{D}=\textrm{
diag}\left(\lambda_{1}(\bm{C}),\lambda_{2}(\bm{C}),\dots,\lambda_{m}(\bm{C})\right)$.
Therefore,
\begin{eqnarray*}
x'\bm{C}x=x'\bm{Q}\bm{Q}'\bm{C}\bm{Q}\bm{Q}'x=y'\bm{D}y=\sum_{i=1}^{m}\lambda_{i}(\bm{C})y_{i}^{2},
\quad{ } y=\bm{Q}'x=\left(y_{1},y_{2},\dots,y_{m}\right)'.
\end{eqnarray*}
Therefore,
\begin{eqnarray*}
|x'\bm{C}x|=|\sum_{i=1}^{m}\lambda_{i}(\bm{C})y_{i}^{2}|\leqslant
\sum_{i=1}^{m} |\lambda_{i}(\bm{C})|y_{i}^{2}\leqslant
\max\{|\lambda_{1}(\bm{C})|,|\lambda_{2}(\bm{C})|,\dots,
|\lambda_{m}(\bm{C})|\}\sum_{i=1}^{m}y_{i}^{2}\\
\leqslant\max\{|\lambda_{1}(\bm{C})|,|\lambda_{2}(\bm{C})|,\dots,|\lambda_{m}(\bm{C})|\}
y'y=\max\{|\lambda_{1}(\bm{C})|,|\lambda_{2}(\bm{C})|,\dots,|\lambda_{m}(\bm{C})|\}
x'x,
\end{eqnarray*}
this completes the proof.
\end{proof}
\begin{cor}\label{corcourant}
Let $\bm{C}$ be a $m\times m$-matrix. Let $x$ and $y$ $m-$column
vectors.
 Then, we have

 $
 \left|x'\bm{C}x\right|\leqslant
\frac{1}{2}\,\max\{|\lambda_{1}(\bm{C}+\bm{C}')|,|\lambda_{2}(\bm{C}+\bm{C}')|,
\dots,|\lambda_{m}(\bm{C}+\bm{C}')|\}x'x$,
and

$ \left|y'\bm{C}x\right|\leqslant
\frac{1}{2}\,\max\{|\lambda_{1}(\bm{B}_{0})|,|\lambda_{2}(\bm{B}_{0})|,\dots,
|\lambda_{2m}(\bm{B}_{0})|\}(x'x+y'y)$,
where
$
\bm{B}_{0}=\left(
              \begin{array}{cc}
                \bm{0} & \bm{C}' \\
                \bm{C} & \bm{0} \\
              \end{array}
            \right).
$
\end{cor}
\begin{proof}
Since $x'\bm{C}x$ is a real number, we have
$x'\bm{C}x=\frac{1}{2}\,x'(\bm{C}+\bm{C}')x$. Therefore, since
$\bm{C}+\bm{C}'$ is a symmetric matrix, the first statement follows
directly from Proposition~\ref{prncourant}. To prove the second
statement, note that $y'\bm{C}x$ can be rewritten as
\begin{eqnarray*}
y'\bm{C}x=(x',y')\left[(\bm{0},\, \bm{C}')'\quad{ } \vdots \quad{ }
\bm{0}\right]\left(x',\,y'\right)'.
\end{eqnarray*}
The rest of the proof follows from the first statement.
\end{proof}
\begin{proof}[Proof of Proposition~\ref{proetaomega}]
Under Assumption~$(\mathcal{H})$, there exists $q_{\scriptscriptstyle 0}>0$ such that\\
$|h(\hat{\beta},\tilde{\beta})|\leqslant q_{\scriptscriptstyle 0}\Big /\|\hat{\beta}-\tilde{\beta}\|^{2}$. Then,
we have
\[
|\eta(h)|\leqslant\textrm{E}\left(|h(\hat{\beta},\tilde{\beta})||\hat{\beta}-\beta|\right)\leqslant q_{\scriptscriptstyle 0} \textrm{E}\left(|\hat{\beta}-\beta|\Big /\|\hat{\beta}-\tilde{\beta}\|^{2}\right),
\]
and then, by using~\eqref{distrbeta} and \eqref{distrU}, we get
\[
\textrm{E}\left(|\hat{\beta}-\beta|/\|\hat{\beta}-\tilde{\beta}\|^{2}\right)
=\textrm{E}\left(|U'_{1}PZ|\,\Big /Z'RZ\right),
\]
this proves the first statement of the proposition. Further, we have
\[
|\omega(h)|=\textrm{E}\left(h^{2}(\hat{\beta},\tilde{\beta})\|\hat{\beta}-\tilde{\beta}\|^{2}\right)\leqslant q^{2}_{\scriptscriptstyle 0} \textrm{E}\left(1\,\Big /\|\hat{\beta}-\tilde{\beta}\|^{2}\right),
\]
and then, by using~\eqref{distrbeta} and \eqref{distrU}, we get
\[
\textrm{E}\left(1\,\Big /\|\hat{\beta}-\tilde{\beta}\|^{2}\right)
=\textrm{E}\left(1\,/Z'RZ\right)=\omega,
\]
this completes the proof.
\end{proof}
\begin{proof}[Proof of Proposition~\ref{born1}]
Let $W=\left(Z',U'_{1}\right)'$. We have
\begin{eqnarray}
\textrm{E}\left\{\left(|U_{1}'\bm{P}Z|\,\big/(Z'\bm{R}Z)\right)
\mathbb{I}_{\left\{\|W\|\leqslant\alpha\right\}}\right\}
=\textrm{E}\left\{\left(|W'\bm{F}W|\,\big /(Z'\bm{R}Z)\right)
\mathbb{I}_{\left\{\|W\|\leqslant\alpha\right\}}\right\},
\end{eqnarray}
where $F$ is defined in~\eqref{BF}. By using
Corollary~\ref{corcourant}, we get
\begin{eqnarray*}
\textrm{E}\left\{\frac{|W'\bm{F}W|}{Z'\bm{R}Z}
\mathbb{I}_{\left\{\|W\|\leqslant\alpha\right\}}\right\}\leqslant
\frac{1}{2}\, \psi_{1}\, \textrm{E}\left\{\frac{1}{Z'\bm{R}Z}
\|W\|^{2}\mathbb{I}_{\left\{\|W\|\leqslant\alpha\right\}}\right\}
\leqslant \frac{1}{2}\, \psi_{1}\,
\textrm{E}\left\{\frac{\alpha^{2}}{Z'\bm{R}Z}
\mathbb{I}_{\left\{\|W\|\leqslant\alpha\right\}}\right\},
\end{eqnarray*}
which gives
\begin{eqnarray}
\textrm{E}\left\{\left(|W'\bm{F}W|\,/(Z'\bm{R}Z)\right)\,
\mathbb{I}_{\left\{\|W\|\leqslant\alpha\right\}}\right\} \leqslant
\alpha^{2}\, \psi_{1}\,
\textrm{E}\left\{1/(Z'\bm{R}Z)\right\}\,\big/2\leqslant \alpha^{2}\,
\psi_{1}\,\omega\,\big /2,
\end{eqnarray}
and the proof is completed.
\end{proof}
\begin{proof}[Proof of Proposition~\ref{born2}]
We have,
\begin{eqnarray*}
\textrm{E}\left\{\frac{|W'\bm{F}W|}{Z'\bm{R}Z}\,
\mathbb{I}_{\left\{\|W\|>\alpha\right\}}\right\}\leqslant
\textrm{E}\left\{\left(|W'\bm{F}W|\,\big /(Z'\bm{R}Z)\right)\,
\mathbb{I}_{\left\{\|Z\|^{2}>\frac{\alpha^{2}}{2}\right\}}\right\}
\leqslant \frac{2}{\alpha^{2}}\textrm{E}\left\{|W'\bm{F}W|\,
\frac{\|Z\|^{2}}{Z'\bm{R}Z}\right\}.
\end{eqnarray*}
By Courant's Theorem, we have the inequality
$$1/\max\{\lambda_{1}(\bm{R}),\lambda_{2}(\bm{R}),\dots,\lambda_{k}(\bm{R})\}
\leqslant\|Z\|^{2}/(Z'\bm{R}Z)\leqslant
 1/\psi_{0},$$ which gives
\begin{eqnarray}
\textrm{E}\left\{\left(|W'\bm{F}W|\,\big /(Z'\bm{R}Z)\right)\,
\mathbb{I}_{\left\{\|Z\|^{2}>\alpha^{2}/2\right\}}\right\} \leqslant
2\,\textrm{E}\left\{|W'\bm{F}W|\right\}\,\big /
(\alpha^{2}\,\psi_{0}).\label{prelident}
\end{eqnarray}
Further, by using Corollary~\ref{corcourant}, we get
\begin{eqnarray}
\textrm{E}\left\{|W'\bm{F}W|\right\}\leqslant
\max\{|\lambda_{1}(\bm{F}+\bm{F}')|,|\lambda_{2}(\bm{F}+\bm{F}')|,\dots,
|\lambda_{2m}(\bm{F}+\bm{F}')|\}\textrm{E}(W'W)\,\big
/2.\label{bornXFX}
\end{eqnarray}
Note that $\bm{F}+\bm{F}'=\bm{B}$ and
$W\sim \mathcal{N}_{2k}\left(\left(
                              \begin{array}{c}
                                \mu \\
                                \bm{0} \\
                              \end{array}
                            \right), \left(
                                       \begin{array}{cc}
                                         \bm{I}_{k} & (\bm{A}-\bm{\Sigma})(\bm{P}^{-1})' \\
                                         \bm{P}^{-1}(\bm{A}-\bm{\Sigma}') & \bm{A} \\
                                       \end{array}
                                     \right)
\right)$.
Then, $\textrm{E}(W'W)=\textrm{trace}(\bm{A})+k+\mu'\mu$, and then,
\begin{eqnarray}
\textrm{E}\left\{|W'\bm{F}W|\right\}\leqslant
\,\psi_{1}\left[\textrm{trace}(\bm{A})+k+\mu'\mu\right]\,\big
/2.\label{impident}
\end{eqnarray}
By combining~\eqref{prelident},~\eqref{bornXFX} and
\eqref{impident}, we get the statement of the proposition,
which completes the proof.
\end{proof}
\begin{prn}\label{propomegah}
Suppose that Assumptions $(\mathcal{H}_{1})$ and $(\mathcal{H}_{2})$ hold. Then, under
normality, $\omega(h)<\infty$ provided that $q\geqslant 3$.
\end{prn}
\begin{proof} We have
$\omega(h)=\textrm{E}\left[h^{2}(\hat{\beta},\tilde{\beta})
\left\|\hat{\beta}-\tilde{\beta}\right\|^{2}\right]$, and then,
under Assumption~$(\mathcal{H}_{1})$, there exist $q_{0}$ such that
\begin{eqnarray*}
 \omega(h)\leqslant q_{0}\textrm{E}\left[1\Big /
\left\|\hat{\beta}-\tilde{\beta}\right\|^{2}\right]\leqslant
q_{0}\textrm{trace}(\bm{\Lambda}\bm{\Xi}\bm{\Lambda})\textrm{E}\left[1\Big/
\left(\hat{\beta}-\tilde{\beta}\right)'\bm{\Lambda}\bm{\Xi}\bm{\Lambda}\left(\hat{\beta}-\tilde{\beta}\right)\right].
\end{eqnarray*}
Further, by using Theorem~5.1.3 in Mathai and Provost~(1992,
p.~199), we have\\
$\left(\hat{\beta}-\tilde{\beta}\right)'\bm{\Lambda}\bm{\Xi}\bm{\Lambda}\left(\hat{\beta}-\tilde{\beta}\right)\sim
\chi^{2}_{q}\left(\gamma'\bm{\Lambda}\bm{\Xi}\bm{\Lambda}\gamma\right)$ and then,
\begin{eqnarray*}
\textrm{E}\left[1\Big/
\left(\hat{\beta}-\tilde{\beta}\right)'\bm{\Lambda}\bm{\Xi}\bm{\Lambda}\left(\hat{\beta}-\tilde{\beta}\right)\right]
=\textrm{E}\left[
\chi^{-2}_{q}\left(\gamma'\bm{\Lambda}\bm{\Xi}\bm{\Lambda}\gamma\right)\right]<1/(q-2).
\end{eqnarray*}
Hence,
$\omega(h)<q_{0}\textrm{trace}(\bm{\Lambda}\bm{\Xi}\bm{\Lambda})/(q-2)<+\infty$ provided that
$q\geqslant 3$, this completes the proof.
\end{proof}
\begin{proof}[Proof of Theorem~\ref{furtherext}]
We have 
\[
\eta(h)=\textrm{E}\left[h(\hat{\beta},\tilde{\beta})\left(\hat{\beta}-\beta\right)'
\left(\hat{\beta}-\tilde{\beta}\right)\right],
\]
and this gives
\[
\eta(h)=
\textrm{E}\left[h(\hat{\beta},\tilde{\beta})\left\|\hat{\beta}-\tilde{\beta}\right\|^{2}\right]
+\textrm{E}\left[h(\hat{\beta},\tilde{\beta})\left(\tilde{\beta}-\beta\right)'
\left(\hat{\beta}-\tilde{\beta}\right)\right].
\]
Then, by the triangular inequality,
\begin{eqnarray*}
 |\eta(h)|\leqslant \left|\textrm{E}\left[h(\hat{\beta},\tilde{\beta})\left\|\hat{\beta}-\tilde{\beta}\right\|^{2}\right]\right|
 +\left|
\textrm{E}\left[h(\hat{\beta},\tilde{\beta})\left(\tilde{\beta}-\beta\right)'
\left(\hat{\beta}-\tilde{\beta}\right)\right]\right|.
\end{eqnarray*}
Then, by using Jensen's inequality and
Assumption~$(\mathcal{H}_{1})$, we get
\begin{eqnarray*}
 |\eta(h)|\leqslant q_{0}+\left|
\textrm{E}\left[h(\hat{\beta},\tilde{\beta})\left(\tilde{\beta}-\beta\right)'
\left(\hat{\beta}-\tilde{\beta}\right)\right]\right|.
\end{eqnarray*}
Further, since $\left(\tilde{\beta}-\beta\right)$ and
$\left(\hat{\beta}-\tilde{\beta}\right)$ are independent, and since,
by Assumption~$(\mathcal{H}_{2})$, $h(\hat{\beta},\tilde{\beta})$ is
a measurable function of $\hat{\beta}-\tilde{\beta}$ only, we have
 $$|\eta(h)|\leqslant q_{0}+\left| \gamma'
\textrm{E}\left[h(\hat{\beta},\tilde{\beta})
\left(\hat{\beta}-\tilde{\beta}\right)\right]\right|.$$
Then, by Cauchy-Schwarz inequality,
\begin{eqnarray*}
 |\eta(h)|\leqslant
q_{0}+\left\|\gamma\right\|\left\{
\textrm{E}\left[h^{2}(\hat{\beta},\tilde{\beta})
\left\|\hat{\beta}-\tilde{\beta}\right\|^{2}\right]\right\}^{1/2}
=q_{0}+\left\|\gamma\right\|\omega^{1/2}(h).
\end{eqnarray*}
Therefore, the proof follows from Proposition~\ref{propomegah}.
\end{proof}


\begin{thebibliography}{100}
 \bibitem{abd} Abdous, B., Genest, C., and R�millard, B.~(2004). Dependence
properties of meta-elliptical distributions. {\it In: Duchesne, P.
and R�millard, B. (Eds.), Statistical Modeling and Analysis for
Complex Data Problems}, Kluwer.

\bibitem{Ash} Ashton, K. G., Burke, R. L., and Layne, J. N. (2007). Geographic variation
in body and clutch of gopher tortoises. \textit{Copeia}, {\bf 2007},
355--363.

\bibitem{bin} Bingham, N. H., and Kiesel, R.~(2001). Semi-parametric modelling in
finance: theoretical foundations. {\it Quantitative Finance}, {\bf
1}, 1--10.

\bibitem{bock1975} Bock, B. E.~(1975). Minimax estimators of the mean of a
multivariate normal distribution. \textit{Annals of Statistics},
{\bf 3}, 209--218.

\bibitem{chmi} Chmielewski, M. A.~(1981). Elliptically Symmetric Distributions: A Review and
Bibliography. {\it International Statistical Review}, {\bf 49}, 1,
67--74.


\bibitem{doug} Douglas, P., and Cobb, C.~(1928).
{\it A theory of production}. American Economic Review, {\bf 18}.

\bibitem{fer} Fernandez-Juricic, E., Sanz, R., and Rodriguez-Prieto, I.~(2003). Testing the
risk-disturbance hypothesis in fragmented landscape: non-linear
responses of house sparrows to humans. {\it Condor}, {\bf 105},
316--326.

\bibitem{fur} Furman, E., and Landsman, Z.~(2006). Tail variance
premium with applications for elliptical portfolio of risks. {\it
Astin Bulletin}, {\bf 36}, 2, 433--462.

\bibitem{gupta} Gupta, A. K., and Varga, T.~(1995). Normal mixture representations of
matrix variate elliptically contoured distributions,{\it
Sankhy$\bar{a}$}, {\bf 57}, 68--78.

\bibitem{ahm2009} Hossain, S., Doksum, K. A., and Ahmed, S. E.~(2009).
Positive shrinkage, improved pretest and absolute penalty estimators
in partially linear models. \textit{Linear Algebra and its
Applications}, {\bf 430}, 2749--2761.

\bibitem{js1961}James, W., Stein, C. (1961). Estimation with quadratic
loss. \textit{ Proc. Fourth Berkeley Symp. Math. Statist. Prob.},
{\bf 1}, 361--379.


\bibitem{judg1978}Judge, G. G., and Bock, M. E. (1978). \textit{The statistical implication of
pre-test and Stein-rule estimators in econometrics.} Amsterdam,
North Holland.

\bibitem{judg2004}Judge, G. G., and Mittelhammer, R. C.~(2004). A Semiparametric Basis for Combining
Estimation Problems under Quadratic Loss. \textit{ Journal of the
American Statistical Association (JASA)}, {\bf 99}, 466, 479--487.

\bibitem{land} Landsman, Z. M., and Valdez, E. A.~(2003).
Tail conditional expectations for elliptical distributions. {\it
North American Actuarial Journal}, {\bf 7}, 4,55--123.

\bibitem{liu} Liu, J. S.,  Cheung, W., and Wong, H.~(2009).
Predictive inference for singular multivariate elliptically
contoured distributions. {\it Journal of Multivariate Analysis},
{\bf 100}, 1440--1446.

\bibitem{mat1992} Mathai, A.M. and Provost, S. B.~(1992). Quadratic Forms in Random Variables
(Statistics: A Series of Textbooks and Monographs). Marcel Dekker,
Inc., New York.

\bibitem{mend} Mendenhall, W., and Sincich, T.~(1992). \textit{Statistics for
Engineering and the Sciences} ($3^{\mbox{rd}}$ ed.). New York:
Dellen Publishing Co.

\bibitem{morri2012} Morris, C. N., and Lysy, M.~(2012). Shrinkage estimation in multilevel
normal models. \textit{Statistical Science}, {\bf 27}, 115--134.

\bibitem{nku} Nkurunziza, S.,  and Ahmed, S. E.~(2010). Shrinkage Drift
Parameter Estimation for Multi-factor Ornstein-Uhlenbeck
Processes.
~{\it Applied Stochastic Models in Business and Industry}, {\bf 26},
2, 103--124.

\bibitem{n2011}Nkurunziza, S. (2011). Shrinkage Strategy In Stratified
Random Sample Subject To Measurement Error. \textit{Statistics and
Probability Letters}, {\bf 81}, 317--325.

\bibitem{ns2013}Nkurunziza, S. (2013). The bias and risk functions of some Stein-rules in elliptically contoured distributions. \textit{Mathematical Methods of Statistics}, {\bf 22}, 1, 70--82.

\bibitem{n2013} Nkurunziza, S., and Chen, F.~(2013). On extension of some identities for the bias and risk functions in
elliptically contoured distributions. \textit{Journal of
Multivariate Analysis}, {\bf 122}, 190--201.

\bibitem{prov}Provost, S. B., and  Cheong, Y. H.~(2000). On the distribution of
linear combinations of the components of a Dirichlet random vector.
{\it Can. J. Stat.}, {\bf 28}, 2, 417--425.

\bibitem{sal} Saleh, A. K. Md.~(2006). {\it Theory of Preliminary Test and
Stein-Type Estimation with Applications}. John Wiley \& Sons, Inc.,
Hoboken, New Jersey.


\bibitem{sal}Tan, Z.~(2015). Improved Minimax Estimation of a
Multivariate Normal Mean under Heteroscedasticity. {\it Bernoulli
Journal}, {\bf 21}, 1, 574--603.

\end{thebibliography}
\end{document}